\documentclass[reqno]{amsart}

\usepackage{amsmath,amssymb,amsthm,amsfonts}
\usepackage{color}
\usepackage{hyperref}
\usepackage{a4wide}

\usepackage{ulem}

\usepackage{graphicx}
\usepackage{caption}
\usepackage{subcaption}

\usepackage{tikz}
\usetikzlibrary{fit,trees,quotes}
\usetikzlibrary{shapes.geometric}

\numberwithin{equation}{section}

\newtheorem{theorem}{Theorem}[section]
\newtheorem{lemma}[theorem]{Lemma}
\newtheorem{prop}[theorem] {Proposition}

\newtheorem{definition}[theorem] {Definition}

\theoremstyle{definition}
\newtheorem{assumption}{Assumption}

\theoremstyle{remark}
\newtheorem{remark}[theorem]{Remark} 
\newtheorem{example}[theorem]{Example}

\newcommand{\e}{\mathrm{e}}
\newcommand{\N}{\mathbb{N}}
\newcommand{\R}{\mathbb{R}}
\newcommand{\Z}{\mathbb{Z}}

\newcommand{\dd}{\mathrm{d}} 
\newcommand{\eps}{\varepsilon}
\newcommand{\la}{\langle}
\newcommand{\ra}{\rangle}

\newcommand{\vect}[1]{\boldsymbol{#1}}

\DeclareMathOperator{\dist}{dist}

\newcommand{\be}{\begin{equation}}
\newcommand{\ee}{\end{equation}}
\newcommand{\ba}{\begin{equation} \begin{aligned}}
\newcommand{\ea}{\end{aligned}\end{equation}}
\newcommand{\bes}{\begin{equation*}}
\newcommand{\ees}{\end{equation*}}

\def\1{{\mathchoice {1\mskip-4mu\mathrm l}      
{1\mskip-4mu\mathrm l}
{1\mskip-4.5mu\mathrm l} {1\mskip-5mu\mathrm l}}}
\newcommand{\heap}[2]{\genfrac{}{}{0pt}{}{#1}{#2}}

\begin{document}

\title{Cluster expansions with renormalized activities and applications to colloids}
\author{Sabine Jansen}
\address{Mathematisches Institut, Ludwig-Maximilians-Universit{\"a}t, 80333 M{\"u}nchen,  Germany}
\email{jansen@math.lmu.de} 
\author{ Dimitrios Tsagkarogiannis}
\address{Dipartimento di Ingegneria e Scienze dell'Informazione e Matematica, Universit\`a degli Studi dell'Aquila, 67100 L'Aquila, Italy}
\email{dimitrios.tsagkarogiannis@univaq.it}
\date{5 March 2019}
\maketitle
\begin{abstract} We consider a binary system of small and large objects in the continuous space interacting via a non-negative potential.
By integrating over the small objects, the effective interaction between the large ones becomes multi-body.
We prove convergence of the cluster expansion for the grand canonical ensemble of the effective system of
large objects. To perform the combinatorial estimate of hypergraphs
(due to the multi-body origin of the interaction) we exploit the underlying structure of the original binary system.
Moreover, we obtain a sufficient condition for convergence which involves the surface of the large objects
rather than their volume. This amounts to a significant improvement in comparison to a direct application of the
known cluster expansion theorems.
Our result is valid for the particular case of hard spheres (colloids) for which we rigorously
treat the depletion interaction.\\


\end{abstract}

\tableofcontents

\section{Introduction}

The present article addresses cluster expansions for binary mixtures made up of ``small'' and ``large'' objects. Our initial motivation 
is the droplet picture of condensation~\cite{hill56book,stillinger1963frenkel-band,sator2003}, where the small objects are molecules of gas and the large objects nascent droplets of liquids or chunks of crystal. Another motivation is the study of colloidal dispersions~\cite{lekkerkerker-tuinier2011book}. Colloids are made of macromolecules with typical size of 1-1000 nanometer dispersed in a medium of much smaller molecules---for example, milk is a colloidal dispersion containing casein micelles (diameter about $200$ nm) dispersed in water~\cite[Chapter 1.1]{lekkerkerker-tuinier2011book}. 
Colloids are best known to mathematical physicists, perhaps, in the context of Brownian motion, first derived to describe the motion of large colloidal particles (e.g., pollen) in the solution~\cite[Chapter 8]{gallavotti1999book}. Devising good thermodynamic models for colloids from first principles is an active area of research in physical chemistry, but to the best of our knowledge, it has attracted little interest in mathematical statistical physics. 

One feature of colloidal systems is that large objects are subject to effective interactions mediated by small particles. A typical phenomenon is depletion attraction~\cite{lekkerkerker-tuinier2011book}. Effective interactions between large objects can be tuned by changing the concentration of the background of small objects, thus opening up intriguing possibilities for the design of new materials. 

Motivated by this point of view we develop a new convergence criterion for a cluster expansion that incorporates the asymmetry between large and small objects. The activity  of large objects is replaced with an effective or ``renormalized'' activity for large objects moving in a sea of small objects. The effective interaction between large objects is obtained by integrating out small objects at pinned positions of the large objects, performing a partial cluster expansion in the activity of small objects. This step bears some resemblance with integrating out a certain length scale in the theory of renormalization group theory~\cite{brydges2007RGlectures}. 
Our bookkeeping is inspired by the mixed partition function investigated by Bovier and Zahradn{\'\i}k~\cite[Section 2]{bovier-zahradnik2000}. 

Our main result is a novel sufficient convergence criterion formulated directly in terms of the activity of small objects and the effective activity of large objects (Theorems~\ref{thm1} and~\ref{thm2}). We apply the theorem to two concrete models, the model of penetrable hard spheres from colloid theory (see Section~\ref{sec:colloids}  below), also called Asakura-Oosawa model~\cite{binder-virnau-statt2014}, and a binary mixture of hard spheres with radii $R > r$ and with respective activities $z_R$ and $z_r$. In the second model the effective activity of large spheres is denoted $\widehat z_R$ and it is a function of $z_R$ and $z_r$. In the binary mixture of hard spheres it behaves, roughly, like 
\be \label{eq:intro1}
	 \widehat z_R = z_R \exp\Bigl( - z_r |B(0,R+r)| + O(z_r^2)\Bigr),
\ee 
with $B(0,R+r)$ the open ball of radius $R+r$ centered at the origin, see Section~\ref{hs} for precision.  In the Asakura-Oosawa model of penetrable hard spheres the $O(z_r^2)$ correction terms in Eq.~\eqref{eq:intro1} vanish, see Section~\ref{sec:colloids}. The effective activity $\widehat z_R$ takes into account the reduction in free volume available to the small spheres. Our convergence criterion works for small activities $z_r$ and effective activities of the order of 
\be \label{eq:intro2}
	 R^d  \widehat z_R \leq \mathrm{const}\,  \exp( - \mathrm{const}\,  z_r |B(0,R+r)\setminus B(0,R-r)|\Bigr),
\ee
see Lemmas~\ref{cor:csuffeasy} and~\ref{lem:suff-hs} below for precisison.

In the limit $R\gg r$ (called colloid limit~\cite[Section 1.3.6]{lekkerkerker-tuinier2011book}), Eq.~\eqref{eq:intro2} reveals two striking features. First, we see that the exponential decay required of the effective activity $\widehat z_R$ goes like $\exp( - \mathrm{const}\, z_r r R^{d-1} )$: it decays with the surface of the large spheres. This should be contrasted with the exponential decay in the volume of the large spheres imposed by Koteck{\'y}-Preiss type convergence conditions, see Proposition~\ref{prop:kpexponential}. Second, rewriting the bound~\eqref{eq:intro2} in terms of the original activity $z_R$ with the help of Eq.~\eqref{eq:intro1}, we see that our convergence condition covers activities $z_R$ that are exponentially large in the volume, 
\be
	R^d  z_R \leq \mathrm{const}\, \exp\Bigl( +\, \mathrm{const}\, z_r |B(0,R+r)|\, \bigl((1 - O( \tfrac{r}{R})\bigr)  + \text{higher order terms} \Bigr),
\ee
again a striking improvement over the exponential decay imposed by Koteck{\'y}-Preiss type conditions. Thus, not only is an expansion with effective activities possible but moreover it leads to drastic improvements over previously available bounds (e.g. \cite{bovier-zahradnik2000}), at least  for the two-scale systems under consideration.

Proofs require us to overcome an impasse: effective interactions between large objects are multi-body, and cluster expansions for multi-body interactions are considerably less developed than for pairwise interactions (see, however, ~\cite{greenberg1971,moraal1976,  procacci-scoppola2000, rebenko2005multibody}). Crucially, the combinatorics
of the multi-body interactions for continuum systems involve hypergraphs and do not permit a direct tree-graph-type inequality that could secure
a convergence result, see the comment at the end of \cite[Appendix B]{brydges1986clustercourse}. To overcome the impasse, we first map the hypergraphs to bipartite graphs (a classical trick in graph theory~\cite{sapozhenk2011hypergraphs}). Capturing the improvements brought about by the switch to effective activities requires further careful considerations, among which the exclusion of graphs with a specific type of articulation point and the choice of an appropriate tree partition scheme that takes into account the asymmetry between large and small objects. Our techniques rely heavily on the concrete form of our effective multi-body interaction; we leave as an open question to which extent our approach may cast light on general multi-body interactions. 

The remaining part of the article is organized as follows. In Section~\ref{sec:colloids} we describe in more detail the penetrable hard spheres model from colloid theory and flesh out the questions resolved in the present article. Section~\ref{sec:main} presents the main results for general binary mixtures with non-negative pair interactions. In section~\ref{sec:partialresum} we perform the partial resummations needed for the definition of the effective activity, effective interactions, and associated representation of the partition function. The key combinatorial estimate underpinning convergence proofs is given in Section~\ref{sec:newtreegraphs}. To conclude, we apply our general theorems to  two concrete models from colloid theory, the penetrable hard spheres model (Section~\ref{sec:appli-colloids}) and the colloid hard sphere model (Section~\ref{hs}). Both are binary mixtures of large and small spheres. In the colloid hard sphere model no two spheres may overlap, while the penetrable hard sphere model small spheres may freely overlap each other but cannot overlap with large spheres. 

\section{Motivation: the penetrable hard sphere-model from colloid theory} \label{sec:colloids} 

Consider a binary mixture of spheres in a box $\Lambda = [0,L]^3\subset \R^3$. The mixture consists of small spheres of radius $r$ and activity $z_r$, and large spheres of radius $R>r$ and activity $z_R$. We may consider the large spheres as colloidal particles moving in a solvent made up of the small spheres. Two large spheres are not allowed to overlap. A large and a small sphere are not allowed to overlap either: a large sphere centered at $x$ creates an excluded volume $B(x,R+r)$, which we may think of as the union of the sphere $B(x,R)$ itself and a ``depletion layer'' $B(x,R+r)\setminus B(x,R)$. 
In order to discuss the main idea we first consider the case of an ideal solvent:
small spheres do not interact between themselves (\emph{penetrable hard spheres} \cite[Chapter 2.1]{lekkerkerker-tuinier2011book}, \underline{Asakura-Oosawa model}~\cite{binder-virnau-statt2014}). In the next section, the main results will be presented in the general case of considering excluded volume for the solvent as well (\emph{colloid hard spheres}~\cite[Chapter 2.3]{lekkerkerker-tuinier2011book}). The grand-canonical partition function is 
\begin{multline} \label{eq:is-partition}
	\Xi_\Lambda(z_R,z_r) = \sum_{n_1,n_2=0}^\infty \frac{z_R^{n_1}}{n_1!} \frac{z_r^{n_2}}{n_2!} \int_{\Lambda^{n_1}}\Bigg\{\int_{\Lambda^{n_2}} \Biggl( \prod_{1\leq i < j \leq n_1} \1_{\{|x_i- x_j|\geq 2R\}}\Biggr) \\
	\times  \Biggl( \prod_{\substack{1\leq i \leq n_1 \\ 1\leq j \leq n_2}} \1_{\{|x_i- y_j|\geq R+r\}}\Biggr)\dd \vect y \Biggr\}
	   \dd \vect x ,
\end{multline} 
where we agree that integrals with zero integration variables are equal to $1$; in particular, the contribution from $n_1=n_2=0$ to the partition function is $1$. The pressure is 
\be
	p(z_R,z_r) = \lim_{L\to \infty}\frac{1}{|\Lambda|} \log \Xi_\Lambda(z_R,z_r). 
\ee
The degrees of freedom related to small spheres can be integrated out explicitly by a computation akin to the Widom-Rowlinson model~\cite{widom-rowlinson1970}. We have
\begin{align} \label{eq:is-integrated}
	\Xi_\Lambda (z_R,z_r) & = \sum_{n_1=0}^\infty \frac{z_R^{n_1}}{n_1!} \int_{\Lambda_1^{n_1}} \Biggl( \prod_{1\leq i < j \leq n_1} \1_{\{|x_i- x_j|\geq 2R\}}\Biggr)  
	\exp\Biggl( z_r \Bigl| \Lambda \setminus \bigcup_{i=1}^{n_1} B(x_i,R+r)\Bigr|\Biggr) \dd \vect x \notag  \\
	& = \e^{z_r|\Lambda|} \sum_{n_1=0}^\infty \frac{z_R^{n_1}}{n_1!} \int_{\Lambda^{n_1}} \Biggl( \prod_{1\leq i < j \leq n_1} \1_{\{|x_i- x_j|\geq 2R\}}\Biggr)  
	\exp\Biggl( -  z_r \Bigl| \Lambda \cap \bigcup_{i=1}^{n_1} B(x_i,R+r)\Bigr|\Biggr) \dd \vect x.
\end{align} 
The volume in the exponential in the first line is the \emph{free volume} available for the depletant particles. In the exponential in the second line of~\eqref{eq:is-integrated}, the intersection with $\Lambda$ only affects particles $x_i$ within a distance smaller than $R+r$ of the boundary.  Dropping it amounts to a change in boundary conditions that becomes irrelevant in the thermodynamic limit. Using inclusion-exclusion, we may write 
\be \label{eq:is-inclusion-exclusion}
	z_r \Bigl|  \bigcup_{i=1}^n B(x_i,R+r)\Bigr| = z_r \sum_{i=1}^n |B(x_i,R+r)| + \sum_{k=2}^n \sum_{1\leq i_1<\cdots <i_k\leq n } W_k (x_{i_1},\ldots,x_{i_k}),
\ee
 where \be \label{eq:wcol}
	W_k(x_{i_1},\ldots,x_{i_k}) = z_r (-1)^{k-1}  \bigl| B(x_{i_1},R+r)\cap \cdots \cap B(x_{i_k},R+r)|
\ee
which depends on $z_r$ and $R$, but for simplicity we do not make it explicit it in the notation. 
Define 
\be \label{eq:zhatcoll}
	\widehat z_R = \widehat z_R(z_R,z_r):= z_R \exp\Bigl( - z_r |B(0,R+r)|\Bigr). 
\ee
Neglecting boundary effects,  we obtain 
\begin{multline} \label{eq:iseffective1}
	\Xi_\Lambda (z_R,z_r) \approx \e^{z_r|\Lambda|} \Biggl( 1+ \sum_{n=1}^\infty \frac{{\widehat z_R}^{n}}{n!} \int_{\Lambda^n} \Biggl( \prod_{1\leq i < j \leq n} \1_{\{|x_i- x_j|\geq 2R\}}\Biggr) \\
	\times  \exp\Biggl( -   \sum_{k=2}^n \sum_{1\leq i_1<\cdots <i_k\leq n } W_k (x_{i_1},\ldots,x_{i_k};z_r) \Biggr) \dd \vect x \Biggr). 
\end{multline}  
It follows that the pressure $p(z_R,z_r)$ is given by the pressure $z_r$ of the ideal solvent plus the pressure of an effective model consisting of large spheres only
\be\label{eq:iseffective2}
	p(z_R,z_r) = z_r + \widehat p(\widehat z_R(z_R,z_r);z_r).
\ee
In the effective model large spheres have  activity $\widehat z_R(z_R,z_r)$ and are subject to additional effective multi-body interactions $W_k$. Notice that the pair interaction
 $W_2(x,y;z_r) = - z_r |B(x,R+r)\cap B(y,R+r)|\leq 0$ is attractive: this corresponds to the phenomenon of \emph{depletion attraction}~\cite{lekkerkerker-tuinier2011book}. The term $W_2$ is sometimes called \emph{Asakura-Oosawa potential} after~\cite{asakura-oosawa1958}. 
The representation~\eqref{eq:iseffective1} raises several questions: 
\begin{enumerate} 
	\item Is there a Mayer expansion for the effective pressure $\widehat p(\widehat z_R)$ in powers of $\widehat z_R$? The prime difficulty is that the theory of cluster expansions for multi-body interactions is much less developed than that of pair interactions.
		\item Is the radius of convergence of the expansion in powers of $\widehat z_R$ (and possibly $z_r$) 
	larger than in the expansion in terms of the original parameters $z_R$ and $z_r$, at least if $R\gg r$?
	\item Does the approach generalize to non-ideal solvents, i.e., interacting small spheres? 
	\item Does the approach generalize to molecules with flexible shapes as opposed to rigid hard spheres? 
\end{enumerate} 
We shall see that the answer to all four questions is yes.

\begin{remark} \label{rem:effective2body}
When $r/R<\frac23 \sqrt 3 - 1\approxeq 0.15$, the multi-body interactions $W_k$ vanish for $k\geq 3$ \cite[p. 118]{lekkerkerker-tuinier2011book} and the effective interaction between colloid hard spheres, incorporating the direct hard-core interaction, is given by 
\be
	W_2^\mathrm{eff}(x,y) = \begin{cases}
					\infty, & \quad |x-y|< 2R,\\
					- z_r |B(x,R+r)\cap B(y,R+r)|, &\quad 2 R\leq  |x-y|<2R+2r,\\
					0, &\quad |x-y|\geq R+r. 
				\end{cases} 
\ee
The interaction is minimal at $|x-y|=2R$, at which point the overlap of depletion layers has volume 
\be
	V_\mathrm{ov}= \frac{2\pi}{3} r^2(3R+2r) = 2\pi R r^2 \bigl(1+ O(\tfrac r R)\bigr) 
\ee
compare Eq.~(2.19a) in~\cite[Chapter 2]{lekkerkerker-tuinier2011book} (with $h=0$ and $\sigma = 2r$). In particular, the effective pair potential is stable and the stability constant is of the order of $R r^2$ (times a constant related to the kissing number). Therefore classical convergence criteria for the Mayer expansion for pair potentials apply. It is instructive to compare the criteria obtained in this way with the bounds that we derive in Section~\ref{sec:appli-colloids}, see Remark~\ref{rem:furtherimprovements} below. 
 For higher values of the ratio $r/R$, additional repulsive forces might be present at larger distances, as discussed in \cite{MCL}.

\end{remark}

\section{Main results} \label{sec:main} 

Let $(\mathbb X,\mathcal X)$ be a measurable space, $v:\mathbb X\times \mathbb X\to \R_+\cup \{\infty\}$ a measurable, non-negative, symmetric function (i.e., $v(x,y) = v(y,x)$), and $\mu$ a measure on $(\mathbb X,\mathcal X)$. Define $f(x,y):= \exp(- v(x,y)) - 1$, with the convention $\exp( - \infty) =0$. Consider the partition function 
\be \label{eq:zx}
	Z_\mathbb X :=1 + \sum_{n=1}^\infty \frac{1}{n!} \int_{\mathbb X^n} \prod_{1\leq i<j\leq n} (1 + f(x_i,x_j)) \dd \mu (x_1)\cdots \dd \mu(x_n). 
\ee
We are interested in the situation where the polymer space $\mathbb X$ has a bipartite structure in terms of ``small'' and ``large'' objects. Thus we assume that $\mathbb X = \mathbb X_\ell\cup \mathbb X_s$ with disjoint measurable subsets $\mathbb X_\ell$ and  $\mathbb X_s$ . We are after an expansion of $\log Z_{\mathbb X_s\cup \mathbb X_\ell}- \log Z_{\mathbb X_s}$ in terms of effective parameters. 

\begin{example}
Our guiding example is the simplistic colloid model from Section~\ref{sec:colloids}, for which we may set 
\be \label{eq:exo1}
	\mathbb X = \Lambda \times \{R,r\},\quad \mathbb X_\ell = \Lambda\times \{R\},\quad \mathbb X_s = \Lambda\times \{r\}.
\ee
The $\sigma$-algebra $\mathcal X$ is the product of the Borel $\sigma$-algebra and the discrete $\sigma$-algebra, and the reference measure $\mu$ is defined by 
\be \label{eq:exo2}
	\mu(A) = \int_\Lambda \1_A(x, R) z_R \,\dd x+ \int_\Lambda \1_A(x,r) z_r \,\dd x\qquad (A\in \mathcal X).
\ee
Mayer's $f$-function is 
\be \label{eq:exo3} 
	f\bigl( (x,a), (y,b)\bigr)  := \begin{cases} 
				- \1_{\{|x-y|< 2 R\}}, &\quad a=b=R,\\
				- \1_{\{|x-y|< R+r\}}, &\quad (a,b) = (r,R)\ \text{or } (R,r),\\
				0, &\quad (a,b) = (r,r).
			\end{cases} 
\ee
\end{example}

\subsection{Preparations I: Effective activity and effective interactions}  

The first step is to integrate out small objects at fixed positions of large objects. This step is analogous to Eqs.~\eqref{eq:is-integrated} and~\eqref{eq:iseffective1} for the hard spheres in the ideal solvent, however because of the interactions between small objects, the partition function for small objects at given positions of the large objects can no longer be computed explicitly and requires already some cluster expansions in the small objects. 
This is given in Proposition~\ref{prop:multibody} below after we introduce some notation.

Recall the \emph{Ursell functions} \cite[Chapter 4.2.2]{ruelle1969book} 
\be \label{eq:ursell} 
	\varphi^\mathsf{T}(x_1,\ldots,x_n) = \sum_{\gamma \in \mathcal C_n} \prod_{\{i,j\} \in E(\gamma)} f(x_i,x_j) \qquad (n \in \mathbb N,\ x_1,\ldots,x_n\in \mathbb X)
\ee
where $\mathcal{C}_n$ is  the collection of connected graphs with vertex set $[n] =\{1,\ldots,n\}$. 
We introduce a new space $\mathbb Y= \sqcup_{k=1}^\infty \mathbb X_s^k$ with signed (or complex, if $\mu$ is complex) measure $\nu$ satisfying 
\be \label{eq:nudef}
	\int_{\mathbb{Y}} h(Y) \dd \nu(Y) = \sum_{k=1}^\infty \frac{1}{k!} \int_{\mathbb X_s^k} h(y_1,\ldots,y_k)  
\varphi^\mathsf T(y_1,\ldots, y_k) \dd \mu^k(y_1,\ldots, y_k) ,
\ee
whenever the integrals and the sum are absolutely convergent. We refer to elements $Y\in \mathbb Y$ as \emph{clouds} or \emph{chains}, compare~\cite{bovier-zahradnik2000}. The interaction between a chain $Y=(y_1,\ldots,y_k)\in \mathbb X_s^k \subset \mathbb Y$ and a large object $x\in \mathbb X_\ell$ is described with 
\be \label{eq:zetadef}
	\zeta (x,Y) =  \zeta\bigl(x,(y_1,\ldots,y_k)\bigr) :=  \prod_{j=1}^k (1 + f(x,y_j)) - 1.
\ee
The effective activity for large objects on a sea of small objects is the measure $\widehat \mu$ on $\mathbb X_\ell$ given by 
\be\label{eq:hatmudef}
	  \dd \widehat \mu(x) := e^{ \int_{\mathbb Y} \zeta(x,Y)  \dd \nu(Y)}\, \dd \mu(x).
\ee
Effective multi-body interactions between large objects are given by 
\be \label{eq:wj}
 - W_{\#J}(\vect x_J) := \int_\mathbb{Y}  \prod_{j\in J} \zeta(x_j,Y) \dd \nu(Y),
\ee
with $J$ a finite non-empty set of cardinality $\#J\geq 2$ and $\vect x_J = (x_j)_{j\in J}\in \mathbb X_\ell^J$. The following assumption guarantees that the quantities are well-defined. 

\begin{assumption} \label{ass:pusmall} 
	The pair interaction $v(x,y)$ is non-negative on $\mathbb X\times \mathbb X$ and for some function $c:\mathbb X\times \mathbb X\to \R_+$, we have 	
	\begin{align}
		\int_{\mathbb X_s} |f(y,y')| \e^{c(y')} \dd |\mu|(y') & \leq c(y), &&(y\in \mathbb X_s) \label{eq:cconv} \\
				\int_{\mathbb X_s} |f(x,y)|\e^{c(y)} \dd|\mu|(y) & < \infty,  &&(x\in \mathbb X_\ell). \label{eq:cconv2}
	\end{align}
\end{assumption}

\begin{lemma} \label{lem:well-defined} 
	Under Assumption~\ref{ass:pusmall}, we have 
	\be
		\int_{\mathbb Y} \prod_{j=1}^k|\zeta(x_j,Y)|\dd |\nu|(Y) <\infty,
	\ee
	for all $k\in \N$ and $x_1,\ldots,x_k\in \mathbb X_\ell$. 
\end{lemma} 

\noindent The lemma is proven in Section~\ref{sec:partialresum}. It is a straightforward consequence of the standard inequalities
\be \label{ineq:zeta}
	- 1\leq \zeta(x,Y) \leq 0,\qquad |\zeta(x, (y_1,\ldots, y_k))|\leq \sum_{j=1}^k |f(x,y_j)|. 
\ee
and~\cite[Theorem 1]{ueltschi2004cluster}. The same theorem says that if in addition to Assumption~\ref{ass:pusmall}, we have 
	\be \label{eq:cfinitevol}
		\int_{\mathbb X_s}\e^{c(x)} \dd |\mu|(x) < \infty,
	\ee
then the partition function for small objects alone is given by 
\be \label{eq:zxsnu}
	\log Z_{\mathbb X_s} = \sum_{k=1}^\infty \frac{1}{k!}\int_{\mathbb X_s^k}\varphi^\mathsf T(y_1,\ldots,y_k)\dd \mu^k(y_1,\ldots, y_k) = \int_{\mathbb Y} 1\,\dd \nu(Y) = \nu(\mathbb Y)
\ee
with absolutely convergent sums and integral. 
  The following proposition takes the place of Eq.~\eqref{eq:iseffective1}.

\begin{prop}[Effective partition function] \label{prop:multibody}
Suppose that Assumption~\ref{ass:pusmall} and the finite-volume condition~\eqref{eq:cfinitevol} hold true. Then we have 
\be\label{multibody}
	\frac{Z_{\mathbb X_s \cup \mathbb X_\ell}}{Z_{\mathbb X_s}} = 1 + \sum_{m=1}^\infty \frac{1}{m!}\int_{\mathbb X_\ell^n} \prod_{1\leq i<j \leq m} (1+f(x_i,x_j)) \exp\Biggl(-\sum_{\heap{J\subset [m]}{\#J\geq 2}} W_{\#J}(\vect x_J) \Biggr)  \dd {\widehat \mu}^m(\vect{x})
\ee
and the integrals and sums entering the definitions of  $\widehat \mu$ in \eqref{eq:hatmudef} and $W_{\#J}$ 
in \eqref{eq:wj} are absolutely convergent. 
\end{prop} 

\noindent The proposition is proven in Section~\ref{sec:partialresum}.
For later purpose we define
\begin{equation}\label{psi}
\psi(x_1,\ldots, x_m):=
\prod_{1\leq i<j \leq m} (1+f(x_i,x_j)) \exp\Biggl(-\sum_{\heap{J\subset [m]}{\#J\geq 2}} W_{\#J}(\vect x_J) \Biggr),
\end{equation}
for the integrand in~\eqref{multibody}, with $\psi(x_1) = 1$, and let $\bigl( \psi^\mathsf T(x_1,\ldots, x_m))_{m\in \N}$ be the uniquely defined family of symmetric functions such that for all $m\in \N$ and $x_1,\ldots, x_m\in \mathbb X_\ell$, we have
\be\label{psitdef}
	\psi(x_1,\ldots, x_m) =: \sum_{r=1}^m \sum_{\{V_1,\ldots,V_r\}\in \mathcal P_m} \psi^\mathsf T\bigl( (x_i)_{i\in V_1}\bigr)\cdots \psi^\mathsf T\bigl( (x_i)_{i\in V_r}\bigr)
\ee
with $\mathcal P_m$ the collection of set partitions of $[m]= \{1,\ldots,m\}$. Later, in Proposition~\ref{prop:psit} we give an explicit formula for $\psi^\mathsf T$.

\begin{remark}
The effective multi-body interaction is stable: we have 
\begin{align} \label{eq:stability} 
	\sum_{\heap{J\subset [m]}{\#J\geq 2}} W_{\#J}(\vect x_J) & = -  \int_{\mathbb Y} \prod_{j=1}^m \bigl( 1+ \zeta(x_j, Y)\bigr) \dd \nu(Y) + \sum_{j=1}^m \int_{\mathbb Y} \zeta(x_j,Y) \dd \nu(Y) \notag \\
	& \geq - m  \sum_{j=1}^m b(x_j)
\end{align}
with $b(x_j)= - \int_{\mathbb Y} \zeta(x_j,Y) \dd \nu(Y)$. 
\end{remark} 

\subsection{Preparations II: Hypergraphs and bipartite leaf-constrained graphs} 
Proposition~\ref{prop:multibody} suggests to look for an expansion of $\log Z_{\mathbb X_s\cup \mathbb X_\ell} - \log Z_{\mathbb X_s}$ in terms of the effective activity. To help motivate the form of expansion coefficients, let us rewrite the integrand in~\eqref{multibody} in a slightly different way: First, in terms of hypergraphs on $m$ vertices and second, with a new class of graphs on a larger set of vertices. 

A \emph{hyperedge} on some underlying set $V$ is a subset $J\subset V$ of cardinality at least $2$; we write $\mathcal E[V]$ for the set of hyperedges on $V$. A \emph{hypergraph} is a pair $\mathfrak h = (V,E)$ consisting of an arbitrary set of vertices $V = V(\mathfrak h)$ and set of hyperedges $E= E(\mathfrak h)\subset \mathcal E[V]$. For $m\in \N$, we write 
\be
	\mathcal E_m:=\{J\subset \{1,\ldots,m\}\mid \#J \geq 2\}. 
\ee
for the set of hyperedges on $\{1,\ldots,m\}$ and $\mathcal H_m$ for the set of hypergraphs with vertices $1,\ldots,m$. 

The contribution from the effective interactions to the integrand in~\eqref{multibody} is 
\be \label{eq:hypergraph}
	\exp \Bigl( - \sum_{\substack{J\subset [m]\\ \#J\geq 2}} W_{\#J}(\vect x_J) \Bigr) 
		= \sum_{\mathfrak h\in \mathcal H_m} \prod_{J\in E(\mathfrak h)} \bigl(\e^{- W_{\#J}(\vect x_J)} - 1\bigr). 
\ee
Ideally, we would like to be able to work directly with such an expression, however we are not able to do so and need to transform hypergraphs into objects amenable to tree-graph inequalities. To that aim we start by expanding the exponentials. Remembering the definition~\eqref{eq:wj} of $W_{\#J}$, we get 
\begin{align}
	\exp \Bigl( - \sum_{\substack{J\subset [m]\\ \#J\geq 2}} W_{\#J}(\vect x_J) \Bigr) 
		& = \sum_{\mathfrak h\in \mathcal H_m} \prod_{J\in E(\mathfrak h)} \Biggl(\sum_{k_{J} = 1}^\infty \frac{1}{k_J!} (- W_{\#J}(\vect x_J))^{k_{J}}\Biggr) \notag \\
		& = \sum_{k=0}^\infty\frac{1}{k!} \   \sum_{\substack{ (k_J)_{J\in \mathcal E_m}\in \N_0^{\mathcal E_m}:\\ \sum_{J\in \mathcal E_m} k_J = k}} \frac{k!}{\prod_{J\in \mathcal E_m} k_J!} \prod_{J\in \mathcal E_m}  \Biggl( \int_{\mathbb Y} \prod_{j\in J} \zeta(x_j, Y) \dd \nu(Y)\Biggr)^{k_{J}}.~\label{eq:multigraph}
\end{align} 
We may think of this expression as a sum over hypergraph with multiple edges, with $k$ the total number of edges and $k_J$ the multiplicity of the hyperedge $J$. The multinomial corresponds to the number of ways to distribute $k$ distinct labels $m+1,\ldots,m+k$ to the hyperedges. Thus~\eqref{eq:multigraph} is a sum over edge-labelled multi-hypergraphs.

Each of the $k$ hyperedges in the edge-labelled multigraph comes with an integral over the space $\mathbb Y$ of chains, so it is natural to switch to a new graph on $m+k$ vertices. The vertices $1,\ldots,m$ represent the large objects with coordinates $q_1,\ldots,q_m$, and vertices $m+1,\ldots,m
+k$ represent chains of small objects. To simplify language a bit, we call the vertices $1,\ldots,m$ \emph{stars} and the vertices $m+1,\ldots,m+k$ \emph{clouds}. A new graph $\gamma'$ (see Figure~\ref{fig1}) is obtained from the underlying edge-labelled multigraph as follows: 
\begin{itemize} 
	\item The graph $\gamma'$ is bipartite:  Edges that link two stars and edges that link two clouds are forbidden. 
	\item An edge $\{s,k\}$ consisting of a star $s$ and a cloud $k$ belongs to $\gamma'$ if and only if in the edge-labelled multigraph, the vertex $s$ belongs to the hyperedge with edge label $k$. 
\end{itemize} 

\begin{figure}
	\begin{subfigure}[b]{0.4\textwidth}
\begin{tikzpicture}[scale=0.6, 
every label/.style ={font=\tiny}, 
vertex/.style={minimum size=2pt,fill,draw,circle},
open/.style={fill=none}, 
cross/.style={minimum size=2pt,fill,draw,circle},
sibling distance=1.5cm,
level distance=1.00cm, 
every fit/.style={ellipse,draw,inner sep=+8pt}, leaf/.style={label={[name=#1]below:$#1$}},auto
]
\draw node [vertex,label={left:$1$}] (x1) at (2,-2) {};
\draw node [vertex,label={right:$2$}] (x2) at (4,0) {};
\draw node [vertex,label={left:$3$}] (x3) at (7,-3) {};
\draw node [vertex,label={left:$4$}] (x4) at (9,-5) {};
\draw node [vertex,label={left:$5$}] (x5) at (8,2) {};
\draw node [vertex,label={left:$6$}] (x6) at (10,3) {};
\node [fit= (x1) (x2) (x3), label={left:$E_1$}] {};
\node [fit= (x3) (x4), label={left:$E_2$}] {};
\node [fit= (x3) (x4), label={right:$E_2$}] {};
\node [fit= (x5) (x6), label={left:$E_3$}] {};
 \end{tikzpicture}
		\caption{Hypergraph consisting of hyperedges $E_1=\{1,2,3\}$, $E_2=\{3,4\}$ (with multiplicity $2$) and $E_3=\{5,6\}$}
		\label{fig1a}
	\end{subfigure}
	\hfill
	\begin{subfigure}[b]{0.4\textwidth}
 \begin{tikzpicture}[scale=0.6,
 every label/.style ={font=\tiny}, 
vertex/.style={minimum size=2pt,fill,draw,circle},
open/.style={fill=none}, 
cross/.style={minimum size=2pt,fill,draw,circle},
sibling distance=1.5cm,
level distance=1.00cm, 
every fit/.style={ellipse,draw,inner sep=+8pt}, leaf/.style={label={[name=#1]below:$#1$}},auto
]
\draw node [vertex,label={left:$1$}] (x1) at (2,-2) {};
\draw node [vertex,label={right:$2$}] (x2) at (4,0) {};
\draw node [vertex,label={left:$3$}] (x3) at (7,-3) {};
\draw node [vertex,label={left:$4$}] (x4) at (9,-5) {};
\draw node [vertex,label={left:$5$}] (x5) at (8,2) {};
\draw node [vertex,label={left:$6$}] (x6) at (10,3) {};
\draw node [vertex, open,label={right:$E_1$}] (x7) at (4.8,-1.9) {};
\draw node [vertex, open,label={right:$E_2$}] (x8) at (8.4,-3.6) {};
\draw node [vertex, open,label={left:$E_2$}] (x9) at (7.6,-4.4) {};
\draw node [vertex, open,label={left:$E_3$}] (x10) at (9,2.5) {};
\draw (x1) edge [" "] (x7) ;
\draw (x2) edge [" "] (x7) ;
\draw (x3) edge [" "] (x7) ;
\draw (x3) edge [" "] (x8) ;
\draw (x3) edge [" "] (x9) ;
\draw (x4) edge [" "] (x8) ;
\draw (x4) edge [" "] (x9) ;
\draw (x5) edge [" "] (x10) ;
\draw (x6) edge [" "] (x10) ;
\end{tikzpicture}
		\caption{Bipartite leaf-constrained graph with black vertices (``stars") and white vertices (``clouds")}
		\label{fig1b}
	\end{subfigure}
	\caption{Example of a hypergraph and its bipartite leaf-constrained graph representation.}
	\label{fig1}
\end{figure}
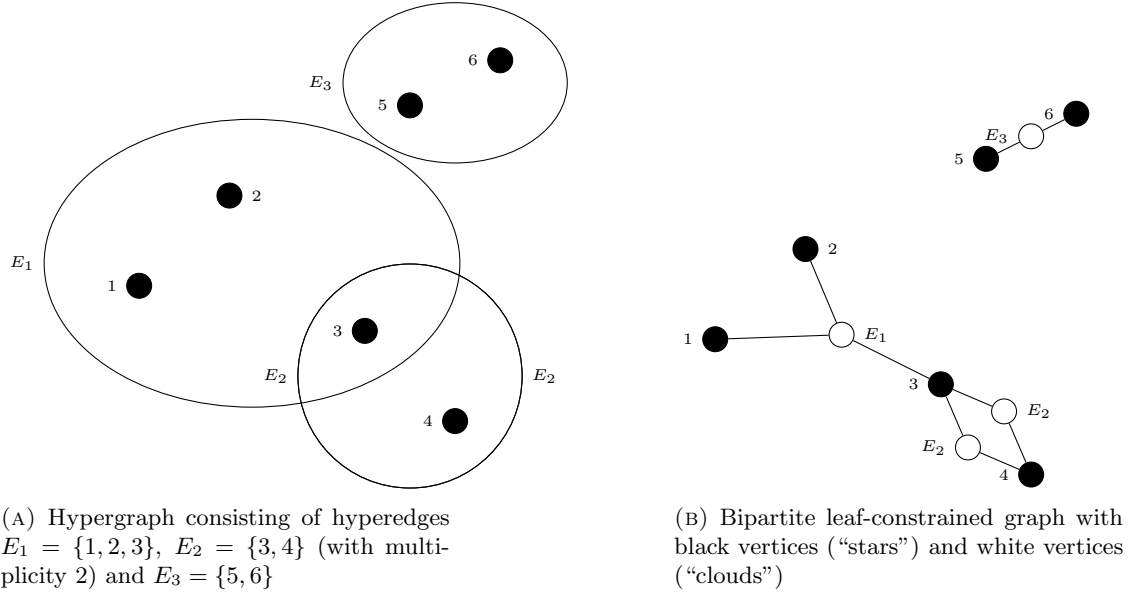

Every cloud is linked to at least two stars because in the original multigraph, every hype
ge comprises at least $2$ vertices. Put differently, clouds cannot be leaves. 

Accordingly let us write $\sum^{m,k}_{\gamma}$ 
for the sum over graphs $\gamma$ with vertices $\{1,\ldots,m+k\}$ such that (i) if $\{i,j\}\in E(\gamma)$, then $i\leq m$ and $j\geq m+1$ (or the other way round), and (ii) every vertex $j\in \{m+1,\ldots,m+k\}$ connects to at least two distinct vertices in $\{1,\ldots,m\}$. Then~\eqref{eq:multigraph} becomes 
\be \label{eq:onourway}
	\exp \Bigl( - \sum_{\substack{J\subset [m]\\ \#J\geq 2}} W_{\#J}(\vect x_J) \Bigr) 
		= \sum_{k=0}^\infty\frac{1}{k!} \int_{\mathbb Y^k} \sideset{}{^{m,k}}\sum_{\gamma}\prod_{\{i,j\}\in E(\gamma)} \zeta(x_i,Y_j) \dd \nu(Y_{m+1})\cdots\dd \nu(Y_{m+k}).
\ee
Finally, in order to compute \eqref{psi}, we need to take into account interactions between large objects, i.e., we have to multiply~\eqref{eq:onourway} by $\prod_{i<j} (1+ f(x_i,x_j))$. This results in a sum over graphs with added links between stars. 

\begin{definition} 	\label{def:cstargraphs}
	For $m,r\in \N_0$ with $m+r \geq 1$, let $\mathcal G_{m,r}^*$ be the class of graphs with vertex set $\{1,\ldots,m+r\}$ such that: 
	\begin{enumerate} 
			\item[(i)] The graph has no edges $\{k_1,k_2\}$ with $k_1,k_2\geq m+1$. 
			\item[(ii)]  Every vertex $k\in \{m+1,\ldots,m+r\}$ belongs to at least two distinct edges $\{s,k\}, \{t,k\}\in E(\gamma)$.
	\end{enumerate} 
 We denote by
	 $\mathcal  C^*_{m,r}:= \mathcal G_{m,r}^* \cap  \mathcal C_{m+r}$ the class of connected graphs that satisfy the same constraints and in addition are connected. Similarly, we denote by 
	 $\mathcal  T^*_{m,r}:= \mathcal G_{m,r}^* \cap  \mathcal T_{m+r}$ the corresponding trees.
\end{definition} 

\noindent 
 In the next proposition we give a formula for the truncated function $\psi^\mathsf T(x_1,\ldots,x_m)$ defined in \eqref{psitdef}. In order to do so, we expand the two factors of \eqref{psi} using \eqref{eq:multigraph} and introduce the
following graph weights :
\begin{equation}\label{omega}	
w_{m,r}(\gamma;x_1,\ldots,x_m, Y_{m+1},\ldots,Y_{m+r}):=
	\Biggl( \prod_{\heap{1\leq s<t\leq m}{\{i,j\}\in E(\gamma)}} f(x_s,x_t)\Biggr) 	
			\Biggl( \prod_{\heap{1\leq s\leq m<k\leq m+r}{\{s,k\}\in E(\gamma)}} \zeta(x_s,Y_k)\Biggr). 
\end{equation}
Note that the effective Boltzmann weight $\psi(x_1,\ldots,x_m)$ is given by a sum over $r\in \N_0$, graphs $\gamma\in \mathcal G^*_{m,r}$, with the $r$ cloud variables integrated over. 
Since the truncated function $\psi^\mathsf T(x_1,\ldots, x_m)$ should contain the ``connected part" of these graphs, we introduce  a modified Ursell function 
\begin{equation}\label{phistar}
\varphi^{\mathsf T}_*(x_1,\ldots,x_m;Y_{m+1},\ldots,Y_{m+r}):= 
	\sum_{\gamma \in \mathcal C^*_{m,r}} w_{m,r}(\gamma;x_1,\ldots,x_m, Y_{m+1},\ldots,Y_{m+r}).
\end{equation}
Furthermore we write $\mathfrak C_m$ for the set of connected hypergraphs with vertex set $\{1,\ldots,m\}$ -- a hypergraph $\mathfrak h $ is connected if for all vertices $i,j$, there exists a sequence of hyperedges $J_1,\ldots, J_p\in E(\mathfrak h)$ such that $i\in J_1$, $j \in J_p$, and $J_r\cap J_{r+1}$ is non-empty for all $r=1,\ldots, p-1$. 

\begin{prop}  \label{prop:psit}
	Under Assumption~\ref{ass:pusmall}, 
for $\psi^\mathsf T$ defined in \eqref{psitdef}
	we have: 
	\begin{align}
		\psi^\mathsf T(x_1,\ldots,x_m) & = \sum_{\mathfrak h \in \mathfrak C_m} \Biggl[\prod_{\substack {1\leq i <j \leq m:\\ \{i,j\}\in E(\mathfrak h)}}  \bigl(  \e^{- v (x_i,x_j) -  W_2(x_i,x_j)} - 1 \bigr) \Biggr]
		\Biggl[\prod_{\substack{J\in E(\mathfrak h): \\  \# J\geq 3}} \bigl( \e^{- W_{\#J}(\vect x_J)} - 1\bigr) \Biggr] \notag \\
		&  = \sum_{r=0}^\infty \frac{1}{r!} \int_{\mathbb Y^r} \varphi^{\mathsf T}_*(x_1,\ldots,x_m;Y_{m+1},\ldots,Y_{m+r}) \dd \nu^r(\vect Y) \label{eq:psit}
	\end{align}
	for all $m\geq 2$ and $x_1,\ldots,x_m\in \mathbb X_\ell^m$ and
	with absolutely convergent sums and integrals. 
\end{prop} 

\begin{proof} 
	The first expression for $\psi^\mathsf T$ follows from the analogous expression for $\psi$ as a sum over not necessarily connected hypergraphs. The computation is similar to~\eqref{eq:hypergraph} but with the original interaction $v(x_i,x_j)$ added to the effective interaction $W_2(x_i,x_j)$ for edges of cardinality $2$. 
	
	In order to go from the first line to the second line, we proceed analogously to the reasoning preceding Definition~\ref{def:cstargraphs} and expand the exponentials. The only subtlety concerns pairwise interactions. Here we write first 
	\be
		  \e^{- v (x_i,x_j) -  W_2(x_i,x_j)} - 1 = f(x_i,x_j) + \bigl( 1+ f(x_i,x_j)\bigr) \bigl( \e^{- W_2(x_i,x_j)} -1 \bigr) 
	\ee
	and then expand $\e^{- W_2(x_i,x_j)}$. Accordingly the original edge $\{i,j\}\in E(\mathfrak h)$ leads to graphs $\gamma$ where the two stars $i$ and $j$ are linked directly, or where they are linked to a common cloud, or both. Adapting the reasoning preceding Definition~\ref{def:cstargraphs}, we arrive at the second line in ~\eqref{eq:psit}. The integrals are absolutely convergent because they come from the underlying integral representations for $W_{\#J}$ (see Lemma~\ref{lem:well-defined}), the sums are absolutely convergent because they come from expansions of the exponential.   
\end{proof}

\subsection{Main theorems}
A final additional ingredient is needed to formulate our main theorems. Our goal is to provide convergence conditions that not only work with the effective activities, but in addition capture improvements brought about by working with effective interactions. A key mechanism, as we shall see, is that we may replace the function $\zeta(x,Y)$ in convergence conditions by a function $\tilde\zeta(x,Y)$ that is smaller in absolute value. 

\begin{example}\label{ex1}
Consider for example our simplistic colloid model from Section~\ref{sec:colloids}. Since small objects do not interact, we may work directly with $\mathbb X_s$ rather than $\mathbb Y$ (in abstract terms, $\nu(\mathbb Y\setminus \mathbb X_s) = 0$). Among the terms to be estimated in cluster expansions, there will be terms of the form
\be
	\bigl( 1+ f(x_1,x_2)\bigr) \zeta(x_1,y) \zeta(x_2,y) = \1_{\{|x_1-x_2| \geq2 R\}} \1_{\{|y- x_1|< R+r\}} \1_{\{|y- x_2|< R+r\}} 
\ee
where by some abuse of notation we identify $x_i\in \Lambda\subset \R^3$ with $(x_i,R)\in \mathbb X_\ell$, similarly for $y$ and $(y,r)$. If $|x_1-x_2|\geq2R$ and $\max ( |y-x_1|, |y-x_2|)< R+r$, then by the triangle inequality 
\be
	|y-x_1|\geq |x_2-x_1|- |y-x_2|>R-r,
\ee
similarly $|y-x_2|>R-r$. Thus if we set 
\be \label{eq:tildezetacoll}
	\tilde \zeta(x,y) = - \1_{\{R-r<|x-y|< R+r\}}
\ee
we get 
\be
	\bigl( 1+ f(x_1,x_2)\bigr) \bigl|\zeta(x_1,y) \zeta(x_2,y)\bigr|  =  \bigl( 1+ f(x_1,x_2)\bigr)\bigl| \tilde \zeta(x_1,y) \tilde \zeta(x_2,y)\bigr|.
\ee
\end{example}
More generally, we assume that the following holds true. 

\begin{assumption} \label{ass:tildezeta}
A function $\tilde\zeta:\mathbb X_\ell\times \mathbb Y\to \R$ is chosen that  satisfies the following conditions: for all $k\geq 2$, $\mu^k$-almost all $(x_1,\ldots,x_k)\in (\mathbb X_\ell)^k$, and $\nu$-almost all $Y\in \mathbb Y$, we have 
\begin{equation}\label{tildezeta1}
	\prod_{1\leq i < j \leq k} (1+f(x_i,x_j)) \prod_{i=1}^k \bigl|\zeta(x_i,Y)\bigr|\leq
			\prod_{1\leq i < j \leq k} (1+f(x_i,x_j)) \prod_{i=1}^k \bigl| \tilde \zeta(x_i,Y)\bigr| 
\end{equation}
and
\begin{equation}\label{tildezeta2}
	\bigl|\zeta(x_1,Y)\bigr| 
			\prod_{1\leq i<j\leq k}(1+f(x_i,x_j))
		\Biggl|\prod_{j=2}^k(1+\zeta(x_j,Y))-1\Biggr|
		\leq
	\bigl|\tilde\zeta(x_1,Y)\bigr| 
		\prod_{1\leq i<j\leq k}(1+f(x_i,x_j)).
\end{equation}
\end{assumption}

\noindent A trivial possible choice is $\tilde\zeta:=\zeta$.
This would correspond to an estimate similar to \cite{bovier-zahradnik2000}.
However, we can do better, see Example~\ref{ex1} above as well as Section~\ref{hs}.

Our main results are:

\begin{theorem}\label{thm1}
Suppose that Assumptions~\ref{ass:pusmall} and~\ref{ass:tildezeta} hold true. 
Let $a: \mathbb X_\ell\to\mathbb R_+$ and $b: \mathbb Y\to\mathbb R_+$ be such that
for $\mu$-almost all $x\in \mathbb X_\ell$ and $\nu$-almost all  $Y\in\mathbb Y$, the following conditions hold:
\begin{align}
\int_{\mathbb Y} |\tilde\zeta(x,Y')|e^{b(y')}\dd|\nu|(Y')
+\int_{\mathbb X_\ell} |f(x,x')|e^{a(x')} \dd|\widehat\mu|(x') & \leq a(x), \label{suff1}\\
\int_{\mathbb X_\ell}|\tilde\zeta(x',Y)|e^{a(x')}\dd|\widehat\mu|(x')
& \leq b(Y)\label{suff2}.
\end{align}
Then, for $\mu$-almost all $x_1\in \mathbb X_\ell$ we have:
\begin{equation}\label{mainest1}
\sum_{m=2}^\infty \frac{1}{(m-1)!}\int_{\mathbb X_\ell^m}
\sum_{r=0}^\infty \frac{1}{r!} \int_{\mathbb{Y}^r}
\left|
\varphi^\mathsf T_*(x_1,\ldots,x_m;Y_1,\ldots,Y_r)	\right| 	
\dd |\nu^{r}|(\vect{Y})
\prod_{i=2}^m\dd |\widehat\mu|(x_i)\leq \e^{a(x_1)}-1,
\end{equation}
while for $\nu$-almost all $Y_1\in\mathbb Y$:
\begin{equation}\label{mainest2}
\sum_{m=2}^\infty \frac{1}{m!}\int_{\mathbb X_\ell^m}
\sum_{r=2}^\infty \frac{1}{(r-1)!} \int_{\mathbb{Y}^r}
\left|
\varphi^\mathsf T_*(x_1,\ldots,x_m;Y_1,\ldots,Y_r)	\right| 	
\prod_{j=2}^r\dd |\nu|(Y_j)
\dd |\widehat\mu^m|(\vect{x})\leq \e^{b(Y_1)}-1 .
\end{equation}
\end{theorem}

\noindent 
The proof is given in Section~\ref{sec:newtreegraphs}. For the expansion of the partition function, we assume in addition that 
\be \label{eq:fivo}
	\int_{\mathbb X} \e^{a(x)} \dd |\widehat \mu|(x) < \infty.
\ee
Remember the functions $\psi^\mathsf T$ from Proposition~\ref{prop:psit}. 

\begin{theorem}  \label{thm2}
	Under the conditions of Theorem~\ref{thm1}, we have
	\be \label{thirda}
		\sum_{m=2}^\infty \frac{1}{(m-1)!} \int_{\mathbb X_\ell^{m-1}} |\psi^\mathsf T(x_1,\ldots, x_m) |\dd |\widehat \mu|(x_2)\cdots \dd |\widehat \mu|(x_m) \leq \e^{a(x_1)} - 1< \infty. 
	\ee
	If in addition condition~\eqref{eq:fivo} holds true, then 
	\be \label{thirdb}
		\log \frac{Z_{\mathbb X_s\cup \mathbb X_\ell}}{Z_{\mathbb X_s} } 	= \sum_{m=1}^\infty \frac{1}{m!} \int_{\mathbb X_\ell^m} \psi^\mathsf T(x_1,\ldots, x_m) \dd \widehat \mu(x_1)\cdots \dd \widehat \mu(x_m) 
	\ee
	with absolutely convergent sums and integrals. 
\end{theorem} 

\begin{proof} 
	The convergence estimate~\eqref{thirda} follows from the representation of $\psi^\mathsf T$ in terms of $\varphi_*^\mathsf T$ given in Proposition~\ref{prop:psit} and the estimate~\eqref{mainest1} in Theorem~\ref{thm1}. The estimate~\eqref{thirda} implies the convergence of the right-hand side of~\eqref{thirdb}. The expression then follows from Proposition~\ref{prop:multibody} and Eq.~\eqref{psitdef}. 
\end{proof}

\section{Partial resummations. Proof of Proposition~\ref{prop:multibody}} \label{sec:partialresum}

\begin{proof}[Proof of Lemma~\ref{lem:well-defined}]
	By Theorem~1 in~\cite{ueltschi2004cluster}, condition~\eqref{eq:cconv} ensures that 
	\be
		\sum_{k=2}^\infty\frac{1}{(k-1)!} \int_{\mathbb X_s^{k-1}} |\varphi^ \mathsf T( x_1,\ldots, x_k)| \dd|\mu|(x_2)\cdots \dd|\mu|(x_k) \leq \e^{c(x_1)} - 1
	\ee
	for all $x_1\in \mathbb X_s$. We combine this estimate with the inequalities~\eqref{ineq:zeta} and obtain
	\begin{align}
		\int_{\mathbb Y} \prod_{j=1}^k|\zeta(x_j,Y)|\dd |\nu|(Y)
				& \leq 		\int_{\mathbb Y} |\zeta(x_1,Y)|\dd |\nu|(Y) \notag  \\
				& \leq \sum_{k=1}^\infty \frac{1}{(k-1)!}  \int_{\mathbb X_s^{k-1}} \bigl|\zeta\bigl(x_1,(y_1,\ldots,y_k)\bigr) \bigr|  |\varphi^ \mathsf T( y_1,\ldots, y_k)| \dd|\mu|(y_1)\cdots \dd|\mu|(y_k) \notag \\
				& \leq \int_{\mathbb X_s} |f(x_1,y_1)| \e^{c(y_1)}  \dd|\mu|(y_1).
	\end{align} 
	The last expression is  finite by assumption~\eqref{eq:cconv2}. 
\end{proof} 

First we prove a lemma. It provides an expression for the partition function analogous to Eq.~\eqref{eq:is-integrated}. 

\begin{lemma} \label{lem:bz}
	Under Assumption~\ref{ass:pusmall}, we have  
	\begin{multline} \label{eq:bz}
		Z_{\mathbb X_s\cup \mathbb X_\ell} = Z_{\mathbb X_s} + \sum_{m=1}^\infty \frac{1}{m!} \int_{\mathbb X_\ell^m} \prod_{1\leq i<j\leq m} (1+f(x_i,x_j))\\
		\times \exp\left(\int_{\mathbb  Y}  \prod_{i=1}^m \bigl(1+\zeta(x_i,Y)\bigr) \dd \nu(Y)\right) \dd \mu^m(\vect x)
	\end{multline} 
	and the integral inside the exponential is absolutely convergent, for all $m\in \N$ and $x_1,\ldots,x_m\in \mathbb X_\ell$. 
\end{lemma} 

\noindent The right-hand side~\eqref{eq:bz} is a ``mixed'' partition function similar to Eq.~(14) in~\cite{bovier-zahradnik2000}.

\begin{proof} 
Starting from the definition~\eqref{eq:zx} of $Z_\mathbb X$, we have 
\be
	Z_\mathbb X = Z_{\mathbb X_s\cup \mathbb X_\ell} = Z_{\mathbb X_s} + \sum_{m=1}^\infty \frac{1}{m!}  \sum_{k=0}^\infty \frac{1}{k!} \int_{\mathbb X_\ell^m\times \mathbb X_s^k} \prod_{1\leq i<j\leq m+k} (1+f(x_i,x_j)) \dd \mu^{n+k}(\vect x)
\ee
which we may rewrite as 
\be \label{eq:zxs}
	Z_{\mathbb X} = Z_{\mathbb X_s} + \sum_{m=1}^\infty \frac{1}{m!}  \int_{\mathbb X_\ell^m}  \prod_{1\leq i<j\leq m} (1+f(x_i,x_j)) 	Z_{\mathbb X_s}^ {\vect x}  \dd\mu^m(\vect x),
\ee
where 
\be
	Z_{\mathbb X_s}^ {\vect x} = \sum_{k=0}^\infty \frac{1}{k!} \int_{\mathbb X_s^k} \prod_{\substack{1\leq i\leq m\\1\leq j\leq k}}\bigl(1+ f(x_i,y_j)\bigr) \prod_{1\leq i<j\leq k} (1+f(y_i,y_j)) \dd \mu^k(\vect y)
\ee
 is the partition function for small objects in the presence of $m$ large objects pinned at the positions $x_1,\ldots, x_m$. It is of the form~\eqref{eq:zx} with $\mathbb X$ replaced by $\mathbb X_s$ and $\dd \mu(x)$ by 
\be
	 \dd \tilde \mu_{\vect x}(y) = \prod_{i=1}^m (1+ f(x_i,y)) \dd \mu(y).
 \ee
Because of $1+f(x_i,y)\leq 1$ by the assumption of non-negative interactions, condition~\eqref{eq:cconv} stays true if we replace $\mu$ with $\tilde \mu_{\vect x}$. Theorem~1 in~\cite{ueltschi2004cluster} then shows that 
\be
	\log Z_{\mathbb X_s}^ {\vect x}  = \sum_{k=1}^\infty \frac{1}{k!}\int_{\mathbb X_s^k} \varphi^\mathsf T(y_1,\ldots, y_k) \prod_{\substack{1\leq i \leq m\\ 1\leq j \leq k}} (1+ f(x_i,y_j)) \dd \mu^k(\vect{y})
\ee
with absolutely convergent series. The definitions of $\zeta$ and $\nu$ yield
\be
	\log Z_{\mathbb X_s}^ {\vect x}  =  \int_{\mathcal Y}  \prod_{i=1}^m \bigl(1+\zeta(x_i,Y)\bigr) \dd \nu(Y).
\ee
The lemma now follows from~\eqref{eq:zxs}. 
\end{proof}

\begin{proof} [Proof of Proposition~\ref{prop:multibody}] 
The proposition is a consequence of Lemma~\ref{lem:bz} and the cluster expansion~\eqref{eq:zxsnu} for $\log Z_{\mathbb X_s}$.  
	First we note that conditions~\eqref{eq:cconv},  \eqref{eq:cfinitevol}, and~\cite[Theorem 1]{ueltschi2004cluster}  ensure the absolute convergence of the expansion~\eqref{eq:zxsnu} for $\log Z_{\mathbb X_s}$ as well as $|\nu|(\mathbb Y) <\infty$. 	
	 Since $|\zeta(x,Y)|\leq 1$ for all $x,Y$, we deduce that $W_{\#J}$ in~\eqref{eq:wj} is given by a convergent integral, 
	 \be
	 	\bigl|W_{\#J}(\vect x_J) \bigr|\leq  \int_\mathbb{Y}  \prod_{j\in J}| \zeta(x_j,Y)| \dd|\nu|( Y) \leq \bigl( |\nu|(\mathbb Y)\bigr)^{\#J}<\infty. 	 
	 \ee
	 The exponent in the definition~\eqref{eq:hatmudef} of $\widehat \mu$ is bounded in a similar way.  So all terms involved are well-defined, and we may rearrange the exponent in Eq.~\eqref{eq:bz} as 
	\be
		\int_{\mathbb Y} \prod_{i=1}^m \bigl(1+\zeta(x_i,Y)\bigr)\dd \nu(Y) = \log Z_{\mathbb X_s} + \sum_{i=1}^n \int_{\mathbb Y} \zeta(x_i,Y)\dd \nu(Y) -\sum_{\heap{J\subset [m]}{\#J\geq 2}} W_{\#J}(\vect x_J),
	\ee
	in analogy with~\eqref{eq:is-inclusion-exclusion}. 
	The proposition follows by dividing both sides in~\eqref{eq:bz} by $Z_{\mathbb X_s}$. 
\end{proof}

\section{Tree-graph inequalities. Proof of Theorem~\ref{thm1}} \label{sec:newtreegraphs}

The idea for the proof of Theorem~\ref{thm1} is to use tree-graph inequalities in combination with a clever choice of partition scheme that takes into account the asymmetry between large objects  and chains of small objects. In the following we refer to large objects as \emph{stars} and chains of small objects as \emph{clouds}.

\begin{definition}[Total order on edges]
	Let $n\in \N$. A total order $\prec$ on the set of edges $\{i,j\}$, $1\leq i<j\leq n$, of the complete graph on $[n]$ is defined as follows: Let $e,e'$ be two edges. Write $e = \{i,j\}$ and $e'=\{i',j'\}$ with $i<j$ and $i'<j'$. Then $e\prec e'$ if and only if either $j'<j$ or $j=j'$ and $i'<i$. 	
\end{definition} 

If we think of each edge $e$ as a two-letter word $ji$, with the order of letters chosen as $j>i$, and an alphabetic ordering of letters such that $n$ precedes $n-1$, etc.,  then the order defined above is a lexicographic order---words are first ordered according to the alphabetical order of their first letter, and then within groups with a common first letter. Thus
\be \label{order} 
	\{n,n\} \prec \{n,n-1\}\prec \cdots \prec \{n,1\} \prec \cdots \prec \{3,3\}\prec \{3,2\}\prec\{3,1\}\prec \{2,2\}\prec\{2,1\} \prec \{1,1\}.  
\ee
We are interested in $n=m+r$ with $m,r\in \N$ and attribute vertices $s\in \{1,\ldots,m\}$ to stars (large objects) and vertices $k\in \{m+1,\ldots,m+r\}$ to clouds (chains of small objects). A key feature of the order~\eqref{order} is that edges that link a cloud to a star are listed before edges that link two stars. 

\begin{definition} [Choice of partition scheme]
	Fix $n\in \N$. For $\tau \in \mathcal T_n$, let $E'(\tau)$ be the collection of edges $\{i,j\}$, $1\leq i <j \leq n$, such that $\{i,j\}\notin E(\tau)$ and every edge $e\in E(\tau)$ on the unique path connecting $i$ to $j$ in $\tau$ is smaller than $\{i,j\}$, i.e., $\{i,j\}\prec e$. Define $R(\tau)\in \mathcal C_n$ as the graph with vertex set $\{1,\ldots,n\}$ and edge set $E(\tau)\cup E'(\tau)$. 
\end{definition} 

It is known that the map $R:\mathcal T_n\to \mathcal C_n$ defined above defines a \emph{tree partition scheme}, i.e., the ``intervals''
\be
	[\tau, R(\tau)]:=\{ \gamma\in \mathcal C_n\mid E(\tau)\subset E(\gamma)\subset E(R(\tau))\}, \quad \tau \in \mathcal T_n
\ee
form a set partition of the connected graphs $\mathcal C_n$. For $\gamma\in \mathcal C_n$, the unique tree $\tau \in \mathcal C_n$ with $\gamma\in [\tau, R(\tau)]$ is reconstructed from $\gamma$ by going through the edges of $\gamma$ in increasing order, keeping an edge if it does not create a loop, discarding it if it does (i.e., following Kruskal's algorithm). 

\begin{lemma} \label{lem:crucialstar}
	Let $m,r\in \N$, $\tau \in \mathcal T_{m+r}$ a tree, $k\in \{m+1,\ldots,m+r\}$ a cloud and $s,s'\in \{1,\ldots,m\}$ two distinct stars not directly linked in $\tau$, i.e., $\{s,s'\}\notin E(\tau)$. 
	\begin{enumerate} 
		\item [(a)]	Assume that the cloud $k$ is directly linked to both stars, i.e., $\{k,s\}$ and $\{k,s'\}\in E(\tau)$. Then the edge $\{s,s'\}$ is necessarily in $E'(\tau)$. 
		\item [(b)] Assume that $\{k,s\}\in E(\tau)$ and $\{k,s'\}\in E'(\tau)$. Then necessarily $\{s,s'\}\in E'(\tau)$. 
		\item [(c)] If $\{k,s\}$ and $\{k,s'\}$ are both in $E'(\tau)$, then so is $\{s,s'\}$. 
	\end{enumerate} 
\end{lemma} 

\begin{proof} 
	(a) The path linking $s$ to $s'$ consists of the two edges $\{s,k\}$ and $\{k,s'\}$. Therefore 
	by definition of the partition scheme $R$, the edge $\{s,s'\}$ is in $R(\tau)$ if and only if
	$\{k,s\}\prec\{s,s'\}$ and $\{k,s'\}\prec \{s,s'\}$.  Now, because of $k>m \geq \max (s,s')$, the edges $\{k,s\}$ and $\{k,s'\}$ are listed in~\eqref{order} before $\{s,s'\}$, and so the required ordering does indeed hold true. 
	
	(b) If $\{k,s'\}$ is in $E'(\tau)$, then by definition of $E'(\tau)$, any edge $e$ on the path connecting $k$ to $s$ is smaller than $\{k,s\}$, i.e., $e\prec \{k,s\}$. 	
	From the definition of the total order of edges, we also know that $\{k,s'\}\prec \{s,s'\}$. Consequently every edge $e$ on the path connecting $k$ to $s$ is also smaller than $\{s,s'\}$. 
	We also have $\{k,s\}\prec \{s,s'\}$. Since the path connecting $s$ to $s'$ consists of the edge $\{s,k\}$ together with the path connecting $k$ to $s'$, we conclude that $\{s,s'\}\in E'(\tau)$. 
	
	(c) Note $\{k,s\}\prec \{s,s'\}$ and $\{k,s'\}\prec \{s,s'\}$. The path connecting $s$ to $s'$ in $E'(\tau)$ consists of the edges on the path connecting $s$ to $k$ and the edges from the path connecting $s'$ to $k$. Any such edge is either smaller than $\{k,s\}$ or smaller than $\{k,s'\}$, so in any case smaller than $\{s,s'\}$. Therefore $\{s,s'\}\in E'(\tau)$. 
\end{proof}

\begin{prop}\label{prop1}
Let $-1\leq \zeta(x,y)\leq 0$ and $-1\leq f(x,x')\leq 0$, $x,x'\in \mathbb X_\ell$, $y\in\mathbb Y$.
We have:
\begin{equation}\label{tg}
\left|
\varphi^\mathsf T_*(x_1,\ldots,x_m; Y_{m+1},\ldots,Y_{m+r})
\right|\leq
\sum_{\tau \in \mathcal T^*_{m+r}} \Bigl( \prod_{\heap{1\leq i<j\leq m}{\{i,j\}\in E(\tau)}} |f(x_i,x_j)|\Bigr) 	
			\Bigl( \prod_{\heap{1\leq i\leq m<j\leq m+r}{\{i,j\}\in E(\tau)}} 
			|\tilde\zeta(x_i,Y_j)|\Bigr) 
\end{equation}
where $\tilde\zeta$ satisfies Assumption~\ref{ass:tildezeta}.
\end{prop}

\begin{proof}
Fix $m,r\in \N$. Remember the graphs $\mathcal C^*_{m,r}\subset \mathcal C_{m+r}$ and the trees 
$\mathcal T^*_{m,r}\subset \mathcal T_{m+r}$ from Definition~\ref{def:cstargraphs}. 
Starting from the definition~\eqref{phistar} of $\varphi_*^\mathsf T$, we have for all $\vect x\in \mathbb X_\ell^m$ and $\vect Y\in \mathbb Y^r$,
\be
	\varphi_*^\mathsf T(\vect x;\vect Y) = \sum_{\gamma \in \mathcal C^*_{m,r}} w(\gamma;\vect x, \vect Y)
		 = \sum_{\tau \in \mathcal T^*_{m+r}}\sum_{\substack{\gamma\in \mathcal C^*_{m,r}: \\
			\gamma\in [\tau,R(\tau)]}} w(\gamma;\vect x, \vect Y).
\ee
Summing over graphs $\gamma\in [\tau, R(\tau)]\cap \mathcal C^*_{m,r}$ amounts to summing over subsets $E'\subset E'(\tau)$, with the understanding that $E(\gamma) = E(\tau)\cup E'$. Any choice of $E'\subset E'(\tau)$ results in a graph $\gamma \in [\tau, R(\tau)]\subset \mathcal C_{m+r}$, however an additional constraint is needed to ensure that $\gamma\in \mathcal C^*_{m,r}$: we need to enforce that every cloud $k\in \{m+1,\ldots,m+r\}$ is directly linked in $\gamma$ to at least two distinct stars $s,t\in \{1,\ldots,m\}$. If a cloud $k$ is already linked to two stars in the tree $\tau$, then we may freely add or not add edges $\{k,t\}\in E'(\tau)$. 
If on the other hand a cloud $k$ is a leaf of the tree, then we need to add at least one edge $\{k,t\}\in E'(\tau)$. Altogether we find 
\begin{align}
	\sum_{\gamma \in \mathcal C^*_{m,r}} w(\gamma;\vect x, \vect Y)
		& =\sum_{\tau \in \mathcal T^*_{m+r}}
			\Biggl( \prod_{\heap{1\leq s<t\leq m}{\{s,t\}\in E(\tau)}} f(x_s,x_t)\Biggr) 	
			\Biggl( \prod_{\heap{1\leq s\leq m<k\leq m+r}{\{s,k\}\in E(\tau)}} 
			\zeta(x_s,Y_k)\Biggr) 
			\nonumber\\
	&\qquad \times
			\Biggl( \prod_{\substack{\{s,t\}\in E'(\tau)\\ 1\leq s< t\leq m}}(1+f(x_s,x_t))\Biggr)
			\prod_{k=m+1}^{m+r}\left[
			\prod_{\substack{t\leq m\\ \{k,t\}\in E'(\tau)}}
			(1+\zeta(x_{t},Y_k))-\1_{\mathcal L(\tau)}(k)
			\right],\label{tg1}
\end{align}
where $\mathcal L(\tau)$ is the set of leaves of $\tau$. 

Standard procedure would have us bound the absolute value of the contribution of edges in $E'(\tau)$, i.e., the second line in~\eqref{tg1}, by~$1$, leading to a bound with the original $\zeta$-functions. Instead we want to replace $\zeta$-terms in the first line of~\eqref{tg1} by $\tilde \zeta$-terms using the assumptions~\eqref{tildezeta1} and~\eqref{tildezeta2} \emph{before} bounding $|1+f|\leq 1$. 

Fix a cloud $k\in \{m+1,\ldots,m+r\}$. Consider first the case that the cloud is linked in $\tau$ to exactly two distinct stars $s,t\leq m$, then by Lemma~\ref{lem:crucialstar}, the edge $\{s,t\}$ belongs to $E'(\tau)$. Therefore the term $1+f(x_s,x_t)$ appears in the second line in~\eqref{tg1}. We combine it with the two terms $\zeta(x_s,Y_k)$ and $\zeta(x_t,Y_k)$ in the first line of~\eqref{tg1}, apply~\eqref{tildezeta1}, and find
\be
	\bigl| (1+f(x_s,x_t)) \zeta(x_s,Y_k)\zeta(x_t,Y_k)\bigr|\leq 	\bigl| (1+f(x_s,x_t)) \tilde \zeta(x_s,Y_k)\tilde \zeta(x_t,Y_k)\bigr|.
\ee
A similar estimate holds true if the cloud $k$ is linked in $\tau$ to three stars or more. The same cloud $k$ also appears in the second line of~\eqref{tg1}; here we simply use the bound 
\be
	0\leq \prod_{\substack{s\leq m:\\ \{k,s\}\in E'(\tau)}}
			(1+\zeta(x_s,Y_k))\leq 1,
\ee
where the non-negativity of the pair potential enters in the form $-1\leq \zeta \leq 0$. 

Next consider the case that the cloud $k$ is a leaf in $\tau$. Then it is linked in $\tau$ to a unique star $s\leq m$. Let 
\be
	\mathcal S_s:=\{ t\in \{1,\ldots,m\} \setminus\{s\} \mid \{k,t\}\in E'(\tau)\}.
\ee
The square bracket in the second line in~\eqref{tg1} becomes 
\be \label{eq:treeskeys}
	\prod_{\substack{t\leq m\\ \{k,t\}\in E'(\tau)}}
		(1+\zeta(x_{t},Y_k))-\1_{\mathcal L(\tau)}(k) 
		= \prod_{t\in \mathcal S_s} (1+\zeta(x_{t},Y_k))-\1_{\mathcal L(\tau)}(k).
\ee 
If $\mathcal S_s$ is empty, i.e., if the set $E'(\tau)$ does not link the cloud to another star $t$, then the right-hand side of ~\eqref{eq:treeskeys} equals $1-1=0$, hence the square bracket in the second line in~\eqref{tg1} vanishes and the tree does not contribute to the sum. If $\mathcal S_s$ is not empty, i.e., if  the set $E'(\tau)$ contains an edge $\{k,t\}$ with $t \leq m$ a star distinct from $s$, then by Lemma~\ref{lem:crucialstar}(b), we must have $\{t,s\} \in E(\tau)$. 
Therefore the term $1+ f(x_s,x_t)$ appears in the first product in the second line of~\eqref{tg1}. Similarly, if $t,t'$ are two distinct stars with $\{k,t\}\in E'(\tau)$ and $\{k,t'\}\in E'(\tau)$, then $\{t,t'\}\in E'(\tau)$ by Lemma~\ref{lem:crucialstar}(c), and the term $1+f(x_t,x_{t'})$ appears in the second line of~\eqref{tg1}.  Eq.~\eqref{tildezeta2} yields 
\begin{multline} 
	\Biggl| \zeta(x_s,Y)\prod_{\substack{1\leq i<j \leq m:\\ i,j\in \{s\}\cup \mathcal S_s}}\bigl( 1+ f(x_i,x_j)\bigr)  \Biggl[ \prod_{\substack{t\leq m:\\  \{k,t\}\in E'(\tau)}}  \bigl( 1 +\zeta(x_t,Y_k)\bigr) -1 \Biggr] \Biggr| \\
		\leq \Biggl| \tilde \zeta(x_s,Y) \prod_{\substack{1\leq i<j \leq m:\\ i,j\in \{s\}\cup \mathcal S_s}}\bigl( 1+ f(x_i,x_j)\bigr) \Biggr|.
\end{multline} 
Combining the above considerations, we obtain 
\begin{multline}
	\Biggl| \sum_{\gamma \in \mathcal C^*_{m,r}} w(\gamma;\vect x, \vect Y)\Biggr|
		 \leq \sum_{\tau \in \mathcal T^*_{m+r}}
			\Biggl| \Biggl( \prod_{\heap{1\leq s<t\leq m}{\{s,t\}\in E(\tau)}} f(x_s,x_t)\Biggr) 	
			\Biggl( \prod_{\heap{1\leq s\leq m<k\leq m+r}{\{s,k\}\in E(\tau)}} 
			\tilde \zeta(x_s,Y_k)\Biggr) \Biggr|\\
		 \times
			\Biggl( \prod_{\substack{\{s,t\}\in E'(\tau)\\ 1\leq s< t\leq m}}(1+f(x_s,x_t))\Biggr).
\end{multline}
To conclude, we use the bound $|1+f(x_i,x_j)|\leq 1$, for $1\leq i < j \leq m$.
\end{proof}

Now we are ready to give the proof of the key convergence estimates. 

\begin{proof}[Proof of Theorem~\ref{thm1}]
In order to prove the bound~\eqref{mainest1} we use the tree-graph inequality from Proposition~\ref{prop1} and find that the left-hand side of \eqref{mainest1} is bounded by:
\begin{equation}\label{pf1}
\sum_{m=2}^\infty \frac{1}{(m-1)!}\int_{\mathbb X_\ell^m}
\sum_{r=0}^\infty \frac{1}{r!} \int_{\mathbb{Y}^r}
\sum_{\tau \in \mathcal T^*_{m,r}}  \prod_{\heap{1\leq i<j\leq m}{\{i,j\}\in E(\tau)}} |f(x_i,x_j)|
			 \prod_{\heap{1\leq i\leq m<j\leq m+r}{\{i,j\}\in E(\tau)}} 
			|\tilde\zeta(x_i,Y_j)|
			\dd |\nu^{r}|(\vect{Y})
\prod_{i=2}^m\dd |\widehat\mu|(x_i).
\end{equation}
For an upper bound, we may go from summation over trees $\tau \in \mathcal T_{m,r}^*$ to summation over trees $\tau \in \mathcal T_{m+r}$. We introduce the abstract polymer space 
$\mathcal P:=\mathbb X_\ell\sqcup \mathbb Y$,
with measure 
\begin{equation}\label{lambda}
	\lambda:=\widehat\mu\oplus\nu
\end{equation}
and
weight function $h:\mathcal P\times\mathcal P\to\mathbb R$ given by
\begin{equation}\label{weight}
	h(P_1,P_2):=
	\begin{cases}
		f(P_1,P_2) & \, P_1,P_2\in \mathbb X_\ell,\\
\tilde\zeta(P_1,P_2) & \, P_1\in \mathbb X_\ell, P_2\in \mathbb Y ,\\
\tilde\zeta(P_2,P_1) & \, P_2\in \mathbb X_\ell, P_1\in \mathbb Y ,\\
0 & \, P_1,P_2\in\mathbb Y.
	\end{cases}
\end{equation}
With this notation, and after substitution of $\mathcal T_{m+r}$ for $\mathcal T^*_{m,r}$, the expression \eqref{pf1} is equal to
\begin{equation}\label{pf2}
	\sum_{n=1}^\infty\frac{1}{(n-1)!}\int_{\mathcal P^{n-1}}\dd |\lambda|(P_2) \cdots \dd |\lambda|(P_n)
\sum_{\tau\in\mathcal T_n}
\prod_{\{i,j\}\in E(\tau)} |h(P_i,P_j)|,
\end{equation}
where $P_1=x_1$.
Similarly, the left-hand side of the inequality~\eqref{mainest2} 
is bounded by  \eqref{pf2} with $P_1=Y_1$.
Proceeding as in the proof of Theorem 2.1 in \cite{poghosyan-ueltschi2009} we obtain that \eqref{pf2}
is bounded by $\e^{a(x_1)}-1$, if $P_1=x_1\in \mathbb X_\ell$, or $\e^{b(Y_1)}-1$, if $P_1=Y_1\in\mathbb Y$.
\end{proof}

\section{Application to penetrable hard spheres}\label{sec:appli-colloids}

In Theorem~\ref{thm1} we have presented new sufficient conditions for the convergence of the cluster expansion.
In order to better appreciate the implied gain we investigate it in the simplest possible model, namely the  penetrable hard spheres presented in Section~\ref{sec:colloids}, where the small spheres  do not interact with each other.

Remember the effective interactions $W_{\#J}$ from~\eqref{eq:wcol} 
\be \label{eq:bmdef}
	b_m(z_r):= \int_{(\R^3)^{m-1}} \psi^\mathsf T(0,x_2,\ldots,x_m) \,\dd x_2\cdots \dd x_m,
\ee
where $\psi^\mathsf T$ is given in \eqref{eq:psit} from which one can also trace back the dependence on $z_r$.
Remember also that $\widehat z_R = \widehat z_R(z_R,z_r) = z_R\exp( - z_r |B(0,R+r)|)$. 

\begin{theorem}  \label{thm:col1} 
Assume that the activities $z_r, z_R>0$ and $\widehat z_R = z_R \exp( - z_r|B(0,R+r)|)$ satisfy 
\begin{align}
	|B(0,2R)| \e^A \widehat z_R +  |B(0,R+r)\setminus B(0,R-r)|\e^a z_r&	\leq A \label{eq:csuff1} \\
						 |B(0, R+r) \setminus B(0,R-r)|\e^A \widehat z_R &\leq a \label{eq:csuff2}
\end{align} 
for some numbers $a,A> 0$. Then 
\be \label{eq:col11}
	p(z_R,z_r) = z_r + \sum_{m=2}^\infty b_m(z_r) {\widehat z}_R^m 
	= z_r + \sum_{m=2}^\infty b_m(z_r) \bigl( z_R\, \e^{- z_r |B(0,R+r)|}\bigr)^m
\ee
with 
\be \label{eq:col12} 
	\sum_{m=1}^\infty| m b_m(z_r)| \widehat z_R ^m \leq \e^A\widehat z_R <\infty.
\ee
\end{theorem} 
\noindent The theorem is proven at the end of this section.

The next lemma presents an even easier sufficient convergence criterion. 

\begin{lemma} \label{cor:csuffeasy}
	For $h>0$ set $\eps(h):=\frac18[(1+h)^3-(1-h)^3]=\frac14(3h+h^3)$. If 
		\be \label{eq:csuffeasy}
				|B(0,2R)| \widehat z_R\leq \frac{1}{\mathrm e} \exp\Bigl(- z_r|B(0,R+r)\setminus B(0,R-r)| \e^{\eps(r/R)}\Bigr)
		\ee
	then the sufficient conditions of Theorem~\ref{thm:col1} are met.
\end{lemma}

\begin{proof}
	 Set $\partial_r B(0,R):=B(0,R+r)\setminus B(0,R-r)$ and 
    $\gamma:=A-z_r|\partial_r B(0,R)|\e^a$. The inequalities~\eqref{eq:csuff1} and~\eqref{eq:csuff2} are equivalent to
    \be \label{eq:csuff3}
    	\begin{aligned}
    			|B(0,2R)|\widehat z_R \exp\Bigl(z_r|\partial_r B(0,R)|\e^a\Bigr) &\leq \gamma \e^{-\gamma},\\
    			 |\partial_r B(0,R)|\widehat z_R \exp\Bigl(z_r|\partial_r B(0,R)|\e^a\Bigr) \e^\gamma &\leq a.
    \end{aligned}			 	
    \ee
    We choose 
    \be
    	 		a:=\eps\bigl(\frac r R)=\frac{|\partial_r B(0,R)|}{|B(0,2R)|},\quad  A:=1+z_r|\partial_r B(0,R)| \e^a
    \ee
    so that $\gamma =1$. For this choice the first inequality in~\eqref{eq:csuff3} reads
    \be \label{eq:csuff4}
    	|B(0,2R)|\widehat z_R \leq \frac{1}{\mathrm e} \exp\Bigl( - z_r|\partial_r B(0,R)|\e^a\Bigr)
    \ee
    and it holds true by the assumption~\eqref{eq:csuffeasy}.   
    For the second inequality in~\eqref{eq:csuff3}, we use~\eqref{eq:csuff4} and estimate 
    \begin{multline}
    		|\partial_r B(0,R)|\widehat z_R \exp\Bigl(z_r|\partial_r B(0,R)|\e^a\Bigr) \e^\gamma \\ = \eps\bigl( \frac rR\bigr)  \Bigl(\mathrm{e} |B(0,2R)|\widehat z_R \exp\Bigl(z_r|\partial_r B(0,R)|\e^a\Bigr) \Bigr) \leq\eps\bigl( \frac rR\bigr)=a.
    \end{multline}
	So we have found that under condition~\eqref{eq:csuffeasy}, there exist $a,A\geq 0$ so that~\eqref{eq:csuff3}  and therefore also~\eqref{eq:csuff1}  and~\eqref{eq:csuff2} hold true.
\end{proof}

\begin{remark}\label{mainremark}
 Our new convergence condition imposes, roughly, that the effective activity decays like $\exp(- \mathrm{const} R^{2}r)$ (in three-dimensional spatial domains), in agreement with the intuition that effective interactions are mediated by the $r$-boundary of large spheres. Moreover, remembering the relation between $\widehat z_R$ and $z_R$, we see that our condition allows for activities $z_R$ that \emph{grow} like
\be \label{eq:boundarydecay}
	|B(0,2R)| z_R \leq \exp\Bigl( z_r |B(0,2R)| \bigl(1-O(\tfrac r R)\bigr) \Bigr).
\ee
\end{remark}
It is instructive to compare the bound~\eqref{eq:boundarydecay} with a direct application of the convergence criterion from ~\cite{ueltschi2004cluster} to the original partition function $Z_{\mathbb X_L}$. The latter criterion asks for the existence of a function $a:\mathbb X_L\to \R_+$ so that 
\be
	\int_{\R^3} |f((q,\sigma),(x,\ell))|\e^{a(x,\ell)} z_R \dd x + \int_{\R^3} |f((q,\sigma),(y,s))|\e^{a(y,s)} z_r \dd y   \leq a(q,\sigma) \qquad ((q,\sigma) \in \mathbb X_L).
\ee	
With the ansatz $a(x)=a$ on $\mathbb X_s$ and $a(x)=A$ on $\mathbb X_\ell$, we have the sufficient convergence conditions
\begin{align}
	|B(0,2R)| \e^A \widehat z_R +  |B(0,R+r)|\e^a z_r&\leq A  \label{eq:u1}\\
						 |B(0, R+r)|\e^A \widehat z_R &\leq a \label{eq:u2}
\end{align} 
that impose $|B(0,2R)| z_R \leq \frac {1}{\mathrm e} \exp(- z_r |B(0,2R)|)$.

\begin{prop} \label{prop:kpexponential}
	If the activities $z_r,z_R\geq 0$ satisfy~\eqref{eq:u1} and~\eqref{eq:u2} for some numbers $a,A\geq 0$, then 
		\be \label{eq:ue0}
			|B(0,2R)|z_R \leq \frac{1}{\mathrm e} \exp\Bigl( -z_r|B(0,R+r)|\Bigr).
		\ee
\end{prop}

\noindent Clearly our condition~\eqref{eq:boundarydecay} is much better than~\eqref{eq:ue0}. 

\begin{proof}
    Eq.~\eqref{eq:u1} implies 
    \be
    		|B(0,2R)|z_R \e^A=
			|B(0,2R)|\Bigr(z_R \e^{z_r|B(0,R+r)| \exp(a)}\Bigr)\e^{A-z_r|B(0,R+r)|\exp(a)}
    		 \leq A-z_r|B(0,R+r)|\e^a 
    \ee
    hence
    \be 
    		|B(0,2R)| z_R \e^{z_r|B(0,R+r)|\exp(a)}\leq \sup_{\gamma \geq 0} \gamma \e^{-\gamma} =\frac {1}{\mathrm e}.
    \ee
    Since $\e^a\geq1$ the proposition follows.
\end{proof}

\begin{remark} \label{rem:furtherimprovements} 
	As noted above, our condition~\eqref{eq:boundarydecay} is much better than~\eqref{eq:ue0}. 
	Remark~\ref{rem:effective2body} together with convergence criteria for attractive pair potentials shows that there is still further room for improvement: When $R$ is much larger than $r$, the effective interaction is a pair potential with stability constant $B$ of order $r^2 R$. The classical convergence criterion~\cite{ruelle1969book}
	\be
			\widehat z_R \, \e^{2 B} 	\int_{\R^d} \bigl|\e^{-W_2^\mathrm{eff}(x)}-1)\bigr| \dd x \leq \frac{1}{\mathrm{e}}
	\ee
	shows that, for $R\gg r$, the expansion in $\widehat z_R$ converges as well for effective activities of order $\widehat z_R \leq \exp( - \mathrm{const}\, R r^2)$, which is better than the decay $\exp( - \mathrm{const} R^2 r)$ imposed by Lemma~\ref{cor:csuffeasy}. We leave as an open question whether such an improved condition can be proven as well in situations where the effective interactions is truly multi-body, e.g.,  for the penetrable hard spheres-model at moderate values of $R/r$ or for the interacting hard spheres treated in Section~\ref{hs}. 
\end{remark}

Next we turn to the expansions of the densities $\rho_R$ and $\rho_r$, defined with the partial derivatives of~ \eqref{eq:col11} as
\be
	\rho_R(z_R, z_r):= z_R\frac{\partial}{\partial z_R} p(z_R,z_r),\quad \rho_r(z_R, z_r):= z_r\frac{\partial}{\partial z_r} p(z_R,z_r). 
\ee

\begin{theorem}\label{thm:dens}
	Under conditions  \eqref{eq:csuff1} and \eqref{eq:csuff2} from Theorem~\ref{thm:col1}, we have 
	\be \label{rhoR1}
		\rho_R (z_R, z_r)  = \sum_{m=1}^\infty m b_m(z_r) \widehat z_R^m
	\ee
	and
	\be\label{rhor}
		\rho_r (z_R,z_r) = z_r \Bigl( 1 - |B(0,R+r)| \rho_R + \sum_{m=2}^\infty \frac{\dd b_m}{\dd z_r}(z_r) \widehat z_R^{m}\Bigr),
	\ee
	with
	\be \label{rhor-estimate}
		 \sum_{m=2}^\infty \Bigl|\frac{\dd b_m}{\dd z_r}(z_r) \widehat z_R^{m} \Bigr|\leq \e^a - 1. 
	\ee	
\end{theorem} 

\noindent The theorem is proven at the end of this section. 

\begin{remark}
A combinatorial representation of $\frac{\dd}{\dd z_r} \psi^\mathsf T(x_1,\ldots,x_m)$ (hence, of $\frac{\dd b_m}{\dd z_r}$) allows for an intuitive interpretation of~\eqref{rhor}. First remember that $\psi^\mathsf T(x_1,\ldots, x_m)$ is given by a sum over connected hypergraphs (see Proposition~\ref{prop:psit}), and the weight of a hypergraph is a product of hyperedge weights. For an edge $E=\{i,j\}$, the derivative of the edge weight is 
\begin{multline} \label{eq:hyper1}
	\frac{\dd}{\dd z_r}\Bigl(  \exp\Bigl(- v(x_i,x_j) - z_r|B(x_i,R+r)\cap B(x_j,R+r)| \Bigr)- 1\Bigr) \\
		 = - |B(x_i,R+r)\cap B(x_j,R+r)|\, \exp\Bigl( - W_2(x_i,x_j)\Bigr). 
\end{multline}
For a hyperedge $J$ with $\# J\geq 3$, the derivative is  
\begin{multline} \label{eq:hyper2}
	\frac{\dd}{\dd z_r}\Bigl( \exp\Bigl( z_r (-1)^{\#J-1} \Bigl|\bigcap_{j\in J}B(x_j,R+r)\Bigr|\Bigr) -1 \Bigr) \\ =  (-1)^{\#J-1} \Bigl|\bigcap_{j\in J}B(x_j,R+r)\Bigr|\, \exp( - W_{\#J}(\vect x_J)). 
\end{multline}
Therefore the derivative of $\psi^\mathsf T(x_1,\ldots, x_m)$ is a sum over weighted hyperedge-rooted hypergraphs $\mathfrak h$. Rooting in an edge $J\in E(\mathfrak h)$ changes the weight of the root edge from the left-hand sides into the right-hand sides of Eqs.~\eqref{eq:hyper1} and~\eqref{eq:hyper2}, respectively. 

On the other hand, since small spheres do not interact, at pinned positions of the large spheres they form an ideal gas. Denote by $\la \cdot \ra$ the expected value with respect to the grand-canonical Gibbs measure. By the ideal gas law for small spheres, the expected number of particles is proportional to the free volume,  left after excluding the hard colloid spheres:
\be \label{eq:idealgas}
	\rho_r |\Lambda| = z_r \la V_\mathrm{free}\ra,
\ee
with the free volume given by inclusion-exclusion as 
\be \label{inclusion-exclusion}
	V_\mathrm{free}=|\Lambda|- N_R |B(0,R+r)|+\sum_{k=2}^{N_R} (-1)^{k-1} \sum_{1\leq i_1<\cdots<i_k\leq N_R} \bigl|B(x_{i_1},R+r)\cap \cdots \cap B(x_{i_k},R+r)\bigr|.
\ee
We insert this expression into ~\eqref{eq:idealgas}, divide by the volume and then let it go to infinity, 
write $ \rho^{(m)}(x_1,\ldots,x_m)$ for the $m$-point correlation function of the large spheres, and obtain (with $x_1:=0$)
\be \label{eq:rhor2}
	\rho_r  = z_r \Bigl( 1 - |B(0,R+r)| \rho_R + \sum_{m=2}^\infty \frac{(-1)^{m-1}}{m!}\int_{(\R^3)^{m-1}} \Bigl|\bigcap_{j=1}^m B(x_j,R+r)\Bigr| \rho^{(m)}(0,x_2,\ldots, x_m) \dd \vect x \Bigr). 
\ee
The first two terms on the right-hand side of~\eqref{eq:rhor2} match the first two terms of~\eqref{rhor}, the third term is reconciled with the help of the combinatorial considerations on the derivative of $\psi^\mathsf T$ (hence, $b_m$) given above. 
\end{remark} 

\begin{remark}
	In the penetrable hard-sphere model, the  estimate~\eqref{rhor-estimate} follows from the convergence bound~\eqref{mainest2} in Theorem~\ref{thm2}. When small spheres interact, additional work is needed. Roughly, this is because the bound~\eqref{mainest2} refers to generating functions for graphs rooted in a cloud while the density of small objects requires rooting in a small sphere. There are many ways to root in a small sphere within a cloud, therefore additional combinatorial factors show up, which we address in future work.
\end{remark}

\begin{proof}[Proof of Theorem~\ref{thm:col1}]
	Theorem~\ref{thm:col1} is deduced  from Theorems~\ref{thm1} and~\ref{thm2} by standard arguments: First we check that convergence conditions are satisfied, uniformly in the volume, then we check that in the cluster expansion we can exchange summation and the infinite-volume limit. For the models we consider, it is well-known that boundary conditions do not affect the thermodynamic limit of the pressure. It is convenient to work with periodic boundary conditions. Define 
	$\Lambda= [-L/2,L/2]^3\subset \R^3$ and set $\mathbb X_L = \Lambda \times \{r, R\}$, $\mathbb X_{L,\ell} = \Lambda \times \{R\}$, $\mathbb X_{L,s} = \Lambda\times \{r\}$.
	By some abuse of notation we  identify $(x,\ell)\in  \mathbb X_{L,\ell}$ with $x\in \Lambda$ and $(y,s) \in \mathbb X_{L,s}$ with $y\in \Lambda$. 
	Let 
	\be
		\mathrm{dist}_L^\mathrm{per}(x,y) = \min_{\vect k \in \Z^3} |x-y- L \vect k|
	\ee
	be the distance with periodic boundary conditions and 
	\be \label{eq:periodicf} 
	\begin{aligned}
		f_L^\mathrm{per}(x,y) & := 
				- \1_{\{\mathrm{dist}_L^\mathrm{per}(x,y) < R+r\}}\quad &(x\in \mathbb X_{L,\ell}), \\
				f_L^\mathrm{per}(x,x') &:= 0 \quad &(x,x'\in \mathbb X_{L,\ell}),\\
				f_L^\mathrm{per}(y,y') &:= 0 \quad &(y,y'\in \mathbb X_{L,s}).
	\end{aligned}
	\ee
	Since small objects do not interact, Assumption~\ref{ass:pusmall} is trivial, moreover chains 
	$Y=(y_1,\ldots,y_k)\in \mathbb Y_L = \sqcup_{k=1}^\infty \mathbb X_{L,s}^k$ of length $k \geq 2$ are irrelevant. In abstract terms, the measure $\nu$ defined in~\eqref{eq:nudef} vanishes on $\mathbb Y_L \setminus \mathbb X_{L,s}$ because the Ursell function $\varphi_k^\mathsf T(y_1,\ldots, y_k)$ vanishes for $k \geq 2$. Eq.~\eqref{eq:zetadef} becomes 
	\be
		\zeta_L^ \mathrm{per}(x,y)= f_L^\mathrm{per}(x,y):= - \1_{\{\mathrm{dist}_L^\mathrm{per}(x,y) < R+r\}}\qquad (x\in \mathbb X_{L,\ell},\, y\in \mathbb X_{L,s}\subset \mathbb Y_L).
	\ee
	Assuming that $L$ is large compared to $R+r$, the effective activity becomes 
	\be
		\dd \widehat \mu_L(x) = \widehat z_{R}\1_\Lambda(x) \dd x,\quad \widehat z_R = z_R \exp\bigl( - z_r |B(0, R+r)|\bigr).
	\ee
	Notice that $\widehat z_R$ does not depend on $L$. 
	Define 
	\be
		\tilde \zeta_{L}^\mathrm{per}(x,y) := - \1_{\{R-r< \mathrm{dist}_{L}^{\mathrm{per}}(x,y)< R+r\}}.
	\ee
	Adapting the considerations from Example~\ref{ex1} to periodic boundary conditions, it is easy to check that $\tilde \zeta_L^\mathrm{per}$ satisfies Assumption~\ref{ass:tildezeta}. Finally let 
	\be \label{abchoice}
		\begin{aligned}
		a(x) &:=A\quad &(x\in \mathbb X_{L,\ell}),\\
		b(y) &:= a\quad  	&(y\in \mathbb X_{L,s}).
	\end{aligned} 
	\ee
	Then $\tilde \zeta_L^\mathrm{per}$, $f_L^\mathrm{per}$ meet the convergence conditions~\eqref{suff1} and~\eqref{suff2}  because of the conditions~\eqref{eq:csuff1} and~\eqref{eq:csuff2}.  Therefore Theorems~\ref{thm1} and~\ref{thm2} are applicable. 
	Define $\psi_{L,\mathrm{per}}^\mathsf T$ is defined just as $\psi^\mathsf T$ but with the effective interactions on $\Lambda$ with periodic boundary conditions. Theorem~\ref{thm2} yields 
	\be \label{eq:zlper}
		\log \frac{Z_{\mathbb X_{L,\ell}\cup\mathbb X_{L,s}}^\mathrm{per}}{Z_{\mathbb X_{L,s}}^\mathrm{per}} = \sum_{m=1}^\infty \frac{\widehat z_R^m}{m!}\int_{\Lambda^m} \psi_{L,\mathrm{per}}^\mathsf{T}(x_1,\ldots, x_m) \dd x_1\cdots \dd x_m
	\ee
	with 
	\be \label{eq:perbounds}
		1 + \sum_{m=2}^\infty \frac{{\widehat z_R}^{m-1}}{(m-1)!} \int_{\Lambda^{m-1}} |\psi_{L,\mathrm{per}}^\mathsf T(x_1,x_2,\ldots,x_m) |\dd x_2\cdots \dd x_m \leq \e^A. 
	\ee
	Next we pass to the limit $L\to \infty$. Using Theorem~\ref{thm1} it is not difficult to see that the bound~\eqref{eq:perbounds} holds true with $\psi^\mathsf T$ instead of $\psi_{L,\mathrm{per}}^\mathsf{T}$. The function $\psi_{L,\mathrm{per}}^\mathsf T$ in general depends on $L$ but it coincides with $\psi^\mathsf T$ when $\min_j \mathrm{dist}(x_j, \partial \Lambda) >R+r$. Replacing $\psi_{L,\mathrm{per}}^\mathsf T$ by $\psi^\mathsf{T}$ in~\eqref{eq:zlper} therefore introduces an error of the order of  $L^{2} (R+r)$ (in dimension three). Indeed, 
	\begin{align} 
	  &\sum_{m=1}^\infty \frac{\widehat z_R^m}{m!}\int_{\Lambda^m}	\bigl|  \psi^\mathsf{T}_{L,\mathrm{per}}(x_1,\ldots, x_m)-\psi^\mathsf{T}(x_1,\ldots, x_m)\bigr| \dd x_1\cdots \dd x_m \notag  \\
	  &\qquad \leq \sum_{m=1}^\infty \frac{\widehat z_R^m}{(m-1)!} \int_{\Lambda^m} \1_{\{\dist(x_1,\partial \Lambda) \leq R+r\}} 
			\bigl(|\psi_{L,\mathrm{per}}^\mathsf T(\vect x)| + |\psi^\mathsf T(\vect x)|\bigr) \dd \vect x \notag \\
	&\qquad \leq  2  \e^A \widehat z_R \bigl(L^3 - (L- R- r)^3\bigr) = O( L^2). 
	\end{align} 
	It follows that 
	\be \label{eq:almost}
		\lim_{L\to \infty}\frac{1}{L^3} 		\log \frac{Z_{\mathbb X_{L,\ell}\cup\mathbb X_{L,s}}^\mathrm{per}}{Z_{\mathbb X_{L,s}}^ \mathrm{per}}
			= \lim_{L\to \infty}  \sum_{m=1}^\infty \frac{\widehat z_R^m}{m!L^3 }\int_{\Lambda^m} \psi^\mathsf{T}(x_1,\ldots, x_m) \dd x_1\cdots \dd x_m. 
	\ee
	By translation invariance, $\int_{\Lambda^m} \psi^\mathsf T (\vect x) \dd \vect x$ is equal to $L^3$ times the integral of $\psi^\mathsf T(0,x'_2,\ldots, x'_m)$ over $(\R^3)^{m-1}$, up to boundary error terms. Exchanging summation and limits in~\eqref{eq:almost} (justified by dominated convergence), we obtain
	\be
		\lim_{L\to \infty}\log \frac{Z_{\mathbb X_{L,\ell}\cup\mathbb X_{L,s}}^\mathrm{per}}{Z_{\mathbb X_{L,s}}^\mathrm{per}} = \sum_{m=1}^\infty \frac{\widehat z_R^m}{m!}   \int_{(\R^3)^{m-1}} \psi^\mathsf T(0,\vect{x'}) \dd \vect {x'}= \sum_{m=1}^\infty \frac{\widehat z_R^m}{m!}  b_m(z_r).
	\ee
	To conclude, we note that $\log Z_{\mathbb X_{L,s}}^\mathrm{per} = z_r$ since small objects  on their own form an ideal gas, and we finally obtain the expansion~\eqref{eq:col11} of the pressure $p(z_R,z_r)$. The bound~\eqref{eq:col12} is an immediate consequence of~\eqref{eq:perbounds} with $\psi^\mathsf T$ instead of $\psi^\mathsf T_{L,\mathrm{per}}$, and of the definition of $b_m(z_r)$. 
\end{proof} 

\begin{proof}[Proof of Theorem~\ref{thm:dens}]
	Let $\widehat p(\widehat z_R, z_r):= \sum_{m=1}^\infty m b_m(z_r) \widehat z_R^m$. Then by Theorem~\ref{thm:col1}, formula~\ref{eq:col11},
	we have $p(z_R,z_r) = z_r + \widehat p(\widehat z_R(z_R, z_r), z_r)$. 
		By the chain rule, we have
	\be \label{rho1}
		\rho_R(z_R,z_r) = z_R \frac{\partial}{\partial z_R} \Bigl( z_r + \widehat p(\widehat z_R (z_R,z_r), z_r) \Bigr) = z_R \frac{\partial \widehat p}{\partial \widehat z_R}( \widehat z_R (z_R,z_r), z_r) \frac{\partial \widehat z_R}{\partial z_R} (z_R,z_r). 
	\ee
	Since $\widehat z_R$ is a linear function of $z_R$, we have
	\be \label{rho2}
		z_R\frac{\partial \widehat z_R}{\partial z_R} (z_R,z_r) = \widehat z_R(z_R,z_r). 
	\ee
	By Theorem~\ref{thm:col1}, the 	partial derivative of $\widehat p$ with respect to its first entry is
	\be \label{rho3}
		\frac{\partial \widehat p}{\partial \widehat z_R} (\widehat z_R,z_r) = \sum_{m=1}^\infty m b_m(z_r)\widehat z_R^{m-1}.
	\ee
	Eq.~\eqref{rhoR1} follows from~\eqref{rho1},~\eqref{rho2} and~\eqref{rho3}. For the density of small spheres, the chain rule yields 
	\be \label{rho4}
		\rho_r(z_R,z_r)
		 =z_r + z_r \frac{\partial \widehat p}{\partial \widehat z_R}( \widehat z_R (z_R,z_r), z_r) \frac{\partial \widehat z_R}{\partial z_r} (z_R,z_r) + z_r \frac{\partial \widehat p}{\partial z_r}\bigl( \widehat z_R(z_R, z_r), z_r\bigr). 
	\ee
	The middle expression is equal to 
	\begin{align}
		&  z_r \frac{\partial \widehat p}{\partial \widehat z_R}( \widehat z_R (z_R,z_r), z_r) \frac{\partial \widehat z_R}{\partial z_r} (z_R,z_r) \notag \\
		 & \qquad = z_r \Bigl( \sum_{m=1}^\infty m b_m(z_r) \widehat z_R(z_R,z_r)^{m-1}\Bigr) \Bigl( - |B(0,R+r)|z_R\, \e^{- |B(0, R+r)|}\Bigr)  \notag \\
		 & \qquad =  - |B(0,R+r)| z_r \rho_R(z_R, z_r). \label{rho5}
	\end{align}
	In the last line we have used~\eqref{rhoR1}. For the partial derivative of $\frac{\partial \widehat p}{\partial z_r}$ in~\eqref{rho4}, assuming we may  exchange differientiation and summation in the definition of $\widehat p$, we get 
	\be \label{rho6}
		z_r \frac{\partial \widehat p}{\partial z_r}\bigl( \widehat z_R(z_R, z_r), z_r\bigr) = z_r\sum_{m=1}^\infty \frac{\dd b_m}{\dd z_r}(z_r) \Bigl( \widehat z_R(z_R, z_r)\Bigr)^{m}. 
	\ee
	In view of $b_1 (z_r) = 1$, the summand for $m=1$ vanishes. 
	The expression~\eqref{rhor} for $\rho_r(z_R, z_r)$ then follows from Eqs.~\eqref{rho4}, \eqref{rho5}, and~\eqref{rho6}. 
	
	It remains to prove the estimate~\eqref{rhor-estimate} (which also justifies the exchange of differentiation and summation leading to~\eqref{rho6}). Because of $\nu (\mathbb Y\setminus \mathbb X_s) =0$, Eq.~\eqref{eq:psit} simplifies and becomes 
	\be
		\psi^\mathsf T(x_1,\ldots, x_m) = \sum_{k=1}^\infty \frac{z_r^k}{k!}\int_{(\R^3)^k} \varphi_*^\mathsf T(x_1,\ldots,x_m;y_1,\ldots, y_k) \dd \vect y, 
	\ee
	which yields 
	\be
		\Bigl|\frac{\dd b_m}{\dd z_r}(z_r) \Bigr|  \leq \sum_{k=1}^\infty \frac{z_r^k}{k!}\int_{(\R^3)^{k+m-1}} |\varphi_*^\mathsf T(0,x_2,\ldots,x_m;y_1,\ldots, y_k)| \dd x_2\cdots \dd x_m \dd y_1\cdots \dd y_k. 
	\ee
	By translation invariance, 
	\be
		\Bigl|\frac{\dd b_m}{\dd z_r}(z_r) \Bigr|  \leq \sum_{k=1}^\infty \frac{z_r^k}{k!}\int_{(\R^3)^{k+m-1}} |\varphi_*^\mathsf T(x_1,x_2,\ldots,x_m;0,y_2,\ldots, y_k)| \dd x_1\cdots \dd x_m \dd y_2\cdots \dd y_k. 
	\ee
	The bound~\eqref{rhor-estimate} now follows from the bound~\eqref{mainest2} in Theorem~\ref{thm1} with the choice~\eqref{abchoice}. 
\end{proof}

\section{Colloid hard sphere model}\label{hs}

Consider a two-type mixture of hard spheres in $\Lambda\subset \R^d$, with two values $R>r>0$ of radii. Define 
\begin{align*}
 1 + f_{\ell\ell}(x_i,x_j) &= \1_{\{ |x_i-x_j|\geq 2R \}},\\
 1 + f_{\ell s}(x_i,y_j) &= \1_{\{ |x_i-y_j|\geq R +r \}}\\
 1 + f_{ss}(y_i,y_j) &= \1_{\{ |y_i-y_j|\geq 2r \}}.
\end{align*}
The counterparts with periodic boundary conditions are denoted $f_{\sigma\tau}^\mathrm{per}(q_i,q_j)$, and for simplicity the volume-dependence is suppressed from the notation (compare~\eqref{eq:periodicf}). 

The grand-canonical partition function in a bounded volume $\Lambda$ with free boundary conditions is 
\begin{multline}\label{tshs}
	\Xi_\Lambda(z_r,z_R) = 
	\sum_{m,k\geq 0}\frac{z_R^m z_r^k}{m!k!}
		\int_{\Lambda^{m+k}}\dd x_1\cdots\dd x_m\, \dd y_{m+1}\cdots \dd y_{m+k}
		\prod_{i <j\leq m}(1+f_{\ell\ell}(x_i,x_j)) \\
		\prod_{i\leq m< j}(1+f_{\ell s}(x_i,y_j))
		\prod_{m< i <j}(1+f_{ss}(y_i,y_j)). 
\end{multline}
We are interested in the expansion of the pressure in finite and infinite volume
\be
	p_\Lambda(z_r,z_R) = \frac{1}{|\Lambda|}\log \Xi_\Lambda(z_r,z_R),
\quad p(z_r,z_R) = \lim_{\Lambda\nearrow \R^3} p_\Lambda(z_r,z_R)
\ee	
and in the finite-volume pressure $p_\Lambda^\mathrm{per}$ with periodic boundary condition in terms of the activity $z_r$ of small spheres and an effective activity $\widehat z_R$ of large spheres. Convergence criteria need to be formulated and checked carefully because the  effective activity depends on the volume and on the boundary conditions.  Recalling \eqref{eq:zetadef} let 
\be \label{zeta-hs}
	 \zeta\bigl(x,(y_1,\ldots,y_k)\bigr) =  \prod_{j=1}^k (1 + f_{\ell s}(x,y_j)) - 1 = - \1_{\{\exists j \in \{1,\ldots, n\}: |x-y_i|< R+r\}}
\ee
be minus the indicator that one of the small spheres centered at $y_1,\ldots,y_n$ overlaps a large sphere centered at $x$. Write $\zeta^\mathrm{per}(x,Y)$ for the counterpart with periodic boundary conditions. Assume that 
\be \label{eq:minconv} 
	|B(0,2r)|\, |z_r| \leq \frac{1}{\mathrm{e}}.
\ee
Define 
\be
	\widehat z_R (x) :=  z_R\, \e^{A(x;z_r)},\quad 	\widehat z_{\Lambda,R}(x) :=  z_R\, \e^{A_\Lambda(x;z_r)},\quad 	\widehat z_{\Lambda,R}^\mathrm{per}(x) := z_{R}\,\e^{A_\Lambda^\mathrm{per}(x;z_r)},
\ee
with 
\begin{align}
	A(x;z_r)&:=  \sum_{n=1}^\infty \frac{z_r^n}{n!} \int_{(\R^d)^n} \sum_{\gamma \in \mathcal C_n}  \zeta( x,(y_1,\ldots, y_n))\sum_{\gamma \in \mathcal C_n}  \prod_{\{i,j\}\in E(\gamma)} f_{ss} (y_i,y_j) \dd \vect y, \\
	A_\Lambda(x;z_r)&:=  \sum_{n=1}^\infty \frac{z_r^n}{n!} \int_{\Lambda^n} \sum_{\gamma \in \mathcal C_n}  \zeta( x,(y_1,\ldots, y_n))\sum_{\gamma \in \mathcal C_n}  \prod_{\{i,j\}\in E(\gamma)} f_{ss}(y_i,y_j) \dd \vect y, \\
	A_\Lambda^\mathrm{per}(x;z_r)&:= \sum_{n=1}^\infty \frac{z_r^n}{n!} \int_{\Lambda^n} \sum_{\gamma \in \mathcal C_n}  \zeta^\mathrm{per}( x,(y_1,\ldots, y_n))\sum_{\gamma \in \mathcal C_n}  \prod_{\{i,j\}\in E(\gamma)} f_{ss} ^\mathrm{per} (y_i,y_j) \dd \vect y.
\end{align} 
The series are convergent by the condition~\eqref{eq:minconv} and standard cluster expansion criteria (compare Lemma~\ref{lem:well-defined}). 
By translation invariance, the effective activities in infinite volume and for periodic boundary conditions are homogeneous, 
\be
		\widehat z_R (x)  = \widehat z_R(0) =:\widehat z_R,\quad \widehat z_{\Lambda,R}^\mathrm{per} (x)  =\widehat z_{\Lambda,R}^\mathrm{per} (0) =: \widehat z_{\Lambda,R}^\mathrm{per}.
\ee
In finite volume, the effective activity can also be expressed as 
\be \label{zhat-finite}
	\widehat z_{\Lambda,R}(x)= \frac{z_R\, \Xi_{\Lambda\setminus B(x,R+r)}(z_r,0)}{\Xi_\Lambda(z_r,0)}.
\ee
The excluded volume is smaller when $x$ is close to the boundary $\partial \Lambda$, accordingly the effective activity $\widehat z_{\Lambda,R}(x)$ is larger. 

Define $b_m(z_r)$ as in~\eqref{eq:bmdef}, and let $b_{\Lambda,m}^\mathrm{per}(z_r)$ be its counterpart with periodic boundary conditions. Let $\partial_r B(0,R):= B(0,R+r) \setminus B(0,R-r)$.

\begin{theorem} \label{thm:hsper}
	Assume that $L > 2( R+r)$ and that the activities $z_r$ and $\widehat z_{\Lambda, R}^\mathrm{per} = \widehat z_{\Lambda, R}^\mathrm{per}(z_r,z_R)$ satisfy 
	\begin{align}
		|B(0,2r)|\,|z_r| \, \e^{b+c}& \leq c,  \label{hsper-suff0} \\
		   |\partial_r B(0,R)|\, |z_r| \e^{b+ c}+ \e^a\,|B(0,2R)||\widehat z_{\Lambda,R}^\mathrm{per}| 	&\leq a, \label{hsper-suff1}\\
	    \e^a\, |\partial_r B(0,R)|\,|\widehat z_{\Lambda,R}^\mathrm{per}|& \leq b, \label{hsper-suff2}
	\end{align}
	for some $a,b,c\geq 0$. Then 
	\be
		p_\Lambda^\mathrm{per}(z_R,z_r) = \sum_{m=1}^\infty b_{\Lambda,m}^\mathrm{per} (z_r)\, (\widehat z_{\Lambda,R}^\mathrm{per})^m,
	\ee
	with 
	\be
		\sum_{m=1}^\infty m \bigl|b_{\Lambda,m}^\mathrm{per} (z_r)\, (\widehat z_{\Lambda,R}^\mathrm{per})^m\bigr| \leq \e^a |\widehat z_{\Lambda,R}^\mathrm{per}|<\infty. 
	\ee
\end{theorem} 

\begin{remark}
	A similar theorem holds true for free boundary conditions, however because of the inhomogeneity of the effective activity, the conditions are written for $\sup_{x\in \Lambda}\widehat z_{\Lambda,R}(x)$. 
\end{remark} 

\begin{proof}
	We check that the sufficient conditions from Theorems~\ref{thm1} and~\ref{thm2} are met. Notice that Eq.~\eqref{hsper-suff1} is formulated directly in terms of the activity $z_r$ of small spheres rather than the activity measure $\nu$ of clouds. First we match the setting of the theorem to the setting studied earlier. We define $\mathbb X, \mathbb X_s, \mathbb X_\ell$  as in~\eqref{eq:exo1}, the activity measure $\mu$ as in~\eqref{eq:exo2},  and we define $f^\mathrm{per}((q,\sigma), (q',\tau)) := f^\mathrm{per}_{\sigma\tau}(q,q')$. For simplicity the volume-dependence is partially suppressed in the notation. The effective activity becomes $\dd \widehat \mu^\mathrm{per}(x) = \widehat z_{\Lambda,R}^\mathrm{per} \dd x$. The measure $\nu^\mathrm{per}$ on $\mathbb Y$ becomes 
	\be
		\int_\mathbb Y h(Y) \dd \nu^\mathrm{per}(Y) = \sum_{n=1}^\infty \frac{z_r^n}{n!} \int_{\Lambda^n} h(y_1,\ldots, y_n) \varphi^{\mathrm{per},\mathsf T}_{ss} (y_1,\ldots,y_n) \dd y_1\cdots \dd y_n,
	\ee
	with $\varphi^{\mathrm{per},\mathsf T}_{ss} (y_1,\ldots,y_n)$ the Ursell function for small spheres. 
	Assumption~\ref{ass:pusmall} is satisfied because of~\eqref{hsper-suff0} and standard cluster expansion criteria ~\cite{ueltschi2004cluster}. For Assumption~\ref{ass:tildezeta}, let 	
	\be
		{\tilde\zeta}^\mathrm{per}(x, (y_1,\ldots, y_n)):= -  \1_{\{\exists i\in \{1,\ldots,n\}:\,R-r < \mathrm{dist}_L^\mathrm{per}(x,y_i) < R+r\}}
	\ee
	be minus the indicator that the cloud $Y$ has at least one point in the ``corona'' of $B(x,R)$:
	\be \label{eq:percorona}
		\{y\in \Lambda\mid  R-r \leq \mathrm{dist}_L^\mathrm{per}(x, y) < R+r\}.
	\ee
	We check that Assumption~\ref{ass:tildezeta} is satisfied.
	Let $k\geq 2$, $x_1,\ldots, x_k \in \mathbb X_s$, and $Y= (y_1,\ldots, y_n) \in \mathbb X_s^n\subset \mathbb Y$. It is enough to consider the case that $Y$ is connected, more precisely, the graph with vertices $\{1,\ldots, n\}$ and edges $\{\{i,j\}\mid\, |y_i-y_j|< 2r \}$ is connected; this is because, for hard-sphere interactions, the set of $Y$'s that do not satisfy this condition is a $\nu$ null set. 
	
	The left-hand side of~\eqref{tildezeta1} is the indicator that $\mathrm{dist}_L^\mathrm{per}(x_i,x_j)\geq  2R$ for all $1\leq i < j \leq k$ and that for each $i \in \{1,\ldots, k\}$, there exists a $j  = j(i) \in \{1,\ldots, n\}$ such that $\mathrm{dist}_L^\mathrm{per}(x_i,y_j) < R+r$. If the indicator vanishes, the inequality~\eqref{tildezeta1} is trival. 
If the indicator is equal to $1$, 
	 there exists $j(1)$ with $\mathrm{dist}_L^\mathrm{per}(x_1,y_{j(1)}) < R+r$. Suppose by contradiction that we cannot impose the additional condition $R-r\leq \mathrm{dist}_L^\mathrm{per}(x_1,y_{j(1)}) <R+r$. Then all points $y_j$ of $Y$ satisfy either $\mathrm{dist}_L^\mathrm{per}(x_1,y_{j}) < R-r$ or $\mathrm{dist}_L^\mathrm{per}(x_1,y_{j}) > R-r$.	
 This splits points $y_j$ into two groups. The first group contains $y_{j(1)}$ and is therefore non-empty. Points between two distinct groups have distance strictly larger than $2r$, so the second group must be empty---otherwise we would have a contradiction with the connectedness of $Y$. Thus all points of $Y$ lie within the ball centered at $x_1$ of radius $R-r$. But this in turn is in contradiction with the existence of $j(2)$ such that $\mathrm{dist}_L^\mathrm{per}(x_2,y_{j(2)}) < R+r$. Thus we have proven that there exists $j(1)$ with $R-r\leq \mathrm{dist}_L^\mathrm{per}(x_1,y_{j(1)}) <R+r$, hence $|\tilde \zeta (x_1,Y)| =1$. A similar argument shows $|\tilde \zeta (x_i,Y)| =1$ for all $i$ and the condition~\eqref{tildezeta1} holds true.

	The left-hand side of~\eqref{tildezeta2} is the indicator that large spheres do not overlap, the cloud $Y$ intersects the large sphere centered at $x_1$ and at least one other large sphere. Proceeding as in the proof of~\eqref{tildezeta1}, we see that on the event that the indicator equals $1$, we must also have $|\tilde \zeta(x_1,Y)| =1$ and we conclude that~\eqref{tildezeta2} holds true, completing the proof that Assumption~\ref{ass:tildezeta} is satisfied.
	
	Finally we turn to the convergence conditions~\eqref{suff1} and~\eqref{suff2}. We define functions $a(\cdot):\mathbb X_\ell \to \R_+$ and $b(\cdot):\mathbb Y\to \R_+$ by 
	\be
		a(x):= a,\qquad b(y_1,\ldots, y_n) := bn	\quad (x\in \mathbb X_\ell,\ Y=(y_1,\ldots, y_n) \in \mathbb Y)
	\ee
	with $a,b\geq 0$ the numbers from Eqs.~\eqref{hsper-suff0}-\eqref{hsper-suff2}. Conditions~\eqref{suff1} and~\eqref{suff2} read 
	\begin{align}
		\int_\mathbb Y|\tilde \zeta^\mathrm{per}(x,Y')| \e^{b(Y')} \dd |\nu^\mathrm{per} |(Y') + \e^a |\widehat z_{\Lambda,R}^\mathrm{per}| |B(0,2R)| &\leq a, \label{persuff1} \\
			 \e^a	|\widehat z_{\Lambda,R}^\mathrm{per}| \int_\Lambda |\tilde \zeta^\mathrm{per}(x',Y)| \dd x' &\leq b(Y). \label{persuff2}
	\end{align} 
	For $Y=(y_1,\ldots, y_n)$, the right-hand side of~\eqref{persuff2} is $bn$ and in the left-hand 
	side we may bound
	\be
		\int_\Lambda |\tilde \zeta^\mathrm{per}(x',Y)| \dd x' \leq 	\int_\Lambda  \sum_{i=1}  ^n \1_{\{R-r < \mathrm{dist}_L^\mathrm{per}(x,y_i) < R+r\}} \dd x' = n |B(0,R+r)\setminus B(0,R-r)|.
	\ee
	Therefore condition~\eqref{hsper-suff2} implies~\eqref{persuff2} hence~\eqref{suff1}. In the left-hand side of~\eqref{persuff1}, we bound 
	\begin{equation}
	\begin{aligned}
		&	\int_\mathbb Y|\tilde \zeta^\mathrm{per}(x,Y')| \e^{b(Y')} \dd |\nu^\mathrm{per} |(Y')\\
		& \qquad  \leq \sum_{n=1}^\infty\frac{|z_r|^n}{n!} \int_{\Lambda^n} \sum_{i=1}^n \1_{\{R-r \leq \mathrm{dist}_L^\mathrm{per}(x,y_i) < R+r\}}\e^{bn} |\varphi_{ss}^{\mathrm{per},\mathsf T} (y_1,\ldots, y_n)| \dd \vect y  \\
		&\qquad = |z_r| \e^b \int_\Lambda  \1_{\{R-r \leq \mathrm{dist}_L^\mathrm{per}(x,y_1) < R+r\}}\Biggl\{\sum_{n=1}^\infty \frac{(|z_r|\e^b)^{n-1}}{(n-1)!} \int_{\Lambda^{n-1}}|\varphi_{ss}^{\mathrm{per},\mathsf T} (y_1,\ldots, y_n)| \dd y_2\cdots \dd y_n \Biggr\} \dd y_1\\
		&\qquad \leq |z_r| \e^b |B(0,R+r)\setminus B(0,R-r)| \e^c.
	\end{aligned} 
	\end{equation} 
	For the last line we have used condition~\eqref{hsper-suff0} and standard estimates from cluster expansions~\cite{ueltschi2004cluster}.
	As a consequence, condition~\eqref{persuff1} implies~\eqref{hsper-suff1} hence~\eqref{suff1}.  The finite-volume condition~\eqref{eq:fivo} is true as well. Thus we have checked all conditions in Theorems~\ref{thm1} and~\ref{thm2} and the proof is complete. 
\end{proof} 

\noindent Next we turn to the pressure in infinite volume. 

\begin{theorem}
	Assume that  $z_r$ and $\widehat z_{ R} = \widehat z_{R}(z_r,z_R)$ satisfy 
	\begin{align}
		|B(0,2r)|\,|z_r| \, \e^{b+c}& \leq c,  \label{hs-suff0} \\
		   |\partial_r B(0,R)|\, |z_r| \e^{b+c}+ \e^a\,|B(0,2R)||\widehat z_{R}| 	& < a, \label{hs-suff1}\\
	    \e^a\, |\partial_r B(0,R)|\,|\widehat z_{R}| & < b \label{hs-suff2}
	\end{align}
	for some $a,b,c\geq 0$. Then 
	\be
		p(z_R,z_r) = \sum_{m=1}^\infty b_{m}(z_r)\, (\widehat z_{R})^m
	\ee
	with 
	\be
		\sum_{m=1}^\infty m \bigl|b_{m}(z_r)\, (\widehat z_{R})^m\bigr| \leq \e^a |\widehat z_{\Lambda}|<\infty. 
	\ee
\end{theorem} 

\noindent The only difference with conditions~\eqref{hsper-suff0}---\eqref{hsper-suff1} is that the inequalities~\eqref{hs-suff1} and~\eqref{hs-suff2} is strict.

\begin{proof}
	As $\Lambda\nearrow \R^3$, the effective activity converges:  $\widehat z_{\Lambda,R}^\mathrm{per} \to \widehat z_R$. Hence, conditions ~\eqref{hs-suff0}---\eqref{hs-suff2} guarantee that for sufficiently large $\Lambda$, the conditions of Theorem~\ref{thm:hsper} are met, uniformly in $\Lambda$.  Theorem~\ref{thm:hsper} is deduced by a passage to the limit similar to the second part of the proof of Theorem~\ref{thm:col1}. 
\end{proof} 

Let us provide simplified convergence conditions similar to  Lemma~\ref{cor:csuffeasy}. 

\begin{lemma}\label{lem:suff-hs}
	Fix $z_r\in \mathbb C$ with  $|B(0,2r)|\, |z_r|< 1/\mathrm{e}$. Then there exist $b,c>0$ such that~\eqref{hs-suff0} holds true. Moreover, given $r,z_r,b,c$ there exist $\alpha,\kappa, R_0>0$ such that if $R\geq R_0$ and 	\be \label{finalsuff}
		|\widehat z_R|\leq \kappa\,  \frac{\exp( - \alpha |\partial_r B(0,R)|)}{R^{ \max(1,d-1)}},
	\ee
	then  conditions~\eqref{hs-suff1} and~\eqref{hs-suff2}  hold.
\end{lemma}

\noindent Note that as $R\to \infty$ at fixed $r>0$, we have $|\partial_r B(0,R)| = O(R^{d-1})$, 

\begin{proof}[Proof of Lemma~\ref{lem:suff-hs}]
	Let $b>0$ small enough so that $|B(0,2r)|\, |z_r|\exp(b)\leq 1/\mathrm{e}$. Remembering that $1/\mathrm{e} = \sup_{x>0} x\exp(-x)$, 
	we deduce that $|B(0,2r)|\, |z_r|\exp(b)\leq c\exp(-c)$ for some $c>0$ and condition~\eqref{hs-suff0} holds true. Let us fix a possible choice of $b,c$.
	
	If~\eqref{hs-suff1} holds true with $\widehat z_R\neq 0$, then necessarily $a> |\partial_r B(0,R)|\, |z_r| \e^{b+ c}$. Thus given $b,c,z_r$, let us choose $a:= \alpha |\partial_r B(0,R)|\, |z_r| $ with $\alpha>\e^{b+c}$. 
	Then conditions~\eqref{hs-suff1} and~\eqref{hs-suff2}  hold true if and only if 
	\be
		\e^a\, |\widehat z_R|< \min \Bigl( \bigl(\alpha- \e^{b+c} \bigr) \frac{|\partial_r B(0,R)|}{|B(0,R)|}, \frac{b}{|\partial_r B(0,R)|} \Bigr).
	\ee
	As $R\to \infty$ at fixed $\alpha,b,c$, the minimum on the right-hand side scales as $\min (R^{-1}, R^{1-d})$ and the lemma follows. 
\end{proof} 

\subsubsection*{Acknowledgments} 
The main part of this article was completed when both authors were members of the Department of Mathematics
at the University of Sussex; 
the authors acknowledge the department for the nice atmosphere. 
The authors also wish to acknowledge Stephen J. Tate who contributed at initial stages of this
work.

\bibliographystyle{amsalpha}
\bibliography{colloid}

\end{document}